%% file: main.tex
\documentclass[a4paper,english,numberwithinsect,cleveref,thm-restate]{lipics-v2021}

\usepackage[utf8]{inputenc} 

\title{Computing m-Eternal Domination Number of Cactus Graphs in Linear Time}

\author{V{\'a}clav Bla{\v z}ej}{Faculty of Information Technology, Czech Technical University in Prague, Prague, Czech Republic}{}{https://orcid.org/0000-0001-9165-6280}{}
\author{Jan Maty{\'a}{\v s} K{\v r}i{\v s}{\v t}an}{Faculty of Information Technology, Czech Technical University in Prague, Prague, Czech Republic}{}{https://orcid.org/0000-0001-6657-0020}{}
\author{Tom{\'a}{\v s} Valla}{Faculty of Information Technology, Czech Technical University in Prague, Prague, Czech Republic}{}{https://orcid.org/0000-0003-1228-7160}{}

\authorrunning{V. Bla{\v z}ej, J. M. K{\v r}i{\v s}{\v t}an, and T. Valla}
\Copyright{V{\' a}clav Bla{\v z}ej, Jan M. K{\v r}i{\v s}{\v t}an, and Tom{\' a}{\v s} Valla}

\input{src/fancyProblem}

\usepackage{lineno} 
\usepackage{todonotes}
\nolinenumbers%

\usepackage{amssymb} 
\usepackage{mathtools}
\usepackage{xspace}
\usepackage{thmtools} 
\usepackage{hyperref}
\usepackage{url}
\usepackage{ifthen}
\usepackage[export]{adjustbox}

\usepackage{tikz}
\usetikzlibrary{arrows,positioning,automata}
\usetikzlibrary{calc,scopes,backgrounds,decorations.pathmorphing,fit}

\graphicspath{{./images/}}

\makeatletter
\newcommand{\customlabel}[2]{%
  \protected@write \@auxout {}{\string \newlabel {#1}{{#2}{\thepage}{#2}{#1}{}} }%
  \hypertarget{#1}{#2}
}
\makeatother

\hideLIPIcs

\crefname{observation}{observation}{observations}
\newtheorem{reduction}{Reduction}{}
\crefname{reduction}{reduction}{reductions}
\newtheorem{property}{Property}{}
\crefname{property}{property}{properties}

\newcommand{\CONF}{C}
\newcommand{\CONFS}{\mathbb{C}}
\newcommand{\FF}{\mathbb{F}}

\newcommand{\EDN}{\gamma^{\infty}_\mathrm{m}} 
\newcommand{\EGC}{\Gamma^{\infty}_\mathrm{m}} 
\newcommand{\DN}{\gamma} 
\newcommand{\DG}{\mathcal{D}} 
\newcommand{\TRAN}{\mathcal{T}} 
\newcommand{\EDNinstance}{m-Eternal Domination\xspace}
\newcommand{\EDNproblem}{\textsc{\EDNinstance}\xspace}
\newcommand{\EGCinstance}{m-Eternal Guard Configuration\xspace}
\newcommand{\EGCproblem}{\textsc{\EGCinstance}\xspace}
\newcommand{\LAB}{labelled strategy\xspace}
\newcommand{\LABS}{labelled strategies\xspace}
\newcommand{\vertexStates}{\mathcal{S}}

\newcommand{\PSPACE}{\textsf{PSPACE}\xspace}
\newcommand{\EXPTIME}{\textsf{EXPTIME}\xspace}

\newcommand\MED{m-eternal domination\xspace}

\newcommand{\GCPOS}[3]{\ensuremath{{#1}\;\square_{#2}\;{#3}}}

\DeclareMathOperator{\white}{\ensuremath{\widehat{\tt 0}}}
\DeclareMathOperator{\pink}{\ensuremath{\widehat{\tt 1}}}
\DeclareMathOperator{\red}{\ensuremath{\widehat{\tt 2}}}
\DeclareMathOperator{\connecting}{\ensuremath{\widehat{\tt X}}}
\DeclareMathOperator{\unknown}{\ensuremath{\widehat{\star}}}
\DeclareMathOperator{\col}{\ensuremath{\operatorname{col}}}

\newcommand\putabove[2]{\mathrel{\overset{\makebox[0pt]{\mbox{\normalfont\tiny\sffamily #2}}}{#1}}}

\newcounter{num}
\newcommand\standardReductionLemma[2]{%
  \setcounter{num}{#2}
  Let $G'$ be $G$ after application of Reduction~\ref{reduction-#1}.
  \ifthenelse{\value{num}>1}{%
    $G$ is defended with $#2$ more guards than $G'$.
    }{
    \ifthenelse{\value{num}>0}{%
      $G$ is defended with $1$ more guard than $G'$.
      }{%
      $G$ is defended with the same number of guards as $G'$.
    }%
  }
}

\newcommand{\reductionimage}[2]{\includegraphics[page=#2,scale=1.1,raise=-\dp\strutbox]{reduction/#1.pdf}}
\newcommand{\reductioncase}[2]{\ref{reduction-#1} & \reductionimage{lower}{#2} & \reductionimage{upper}{#2} \\}

\begin{document}

\maketitle

\newcommand\blfootnote[1]{%
  \begingroup
  \renewcommand\thefootnote{}\footnote{#1}%
  \addtocounter{footnote}{-1}%
  \endgroup
}

\blfootnote{The authors acknowledge the support of the OP VVV MEYS funded
project CZ.02.1.01/0.0/0.0/16\_019/0000765 ``Research Center for Informatics''.
This work was supported by the Grant Agency of the Czech Technical University in Prague, grant \mbox{No.~SGS20/208/OHK3/3T/18}.
Supported by the grant 22-19557S of the Czech Science Foundation.}

\begin{abstract}
  In \MED attacker and defender play on a graph.
  Initially, the defender places guards on vertices.
  In each round, the attacker chooses a vertex to attack.
  Then, the defender can move each guard to a neighboring vertex and must move a guard to the attacked vertex.
  The \MED number is the minimum number of guards such that the graph can be defended indefinitely.

  In this paper, we study the \MED number of cactus graphs.
  We consider two variants of the \MED number: one allows multiple guards to occupy a single vertex, the second variant requires the guards to occupy distinct vertices.
  We develop several tools for obtaining lower and upper bounds on these problems and we use them to obtain an algorithm which computes the minimum number of required guards of cactus graphs for both variants of the problem.
\end{abstract}

\section{Introduction}\label{sec:introduction}

Consider the following game, played by an attacker and a defender on graph $G$.
The defender controls a set of guards, which he initially places on the vertices of $G$.
Each vertex can be occupied by at most one guard.

In each round, the attacker first chooses one vertex, which he \emph{attacks}.
The defender then must \emph{defend} against the attack by moving some or all of his guards along their adjacent edges, so that one of the guards moves to the attacked vertex.

If the attacked vertex is not occupied by a guard after the attack, the attacker wins.
The defender wins if he can defend indefinitely.

Defending a graph from attacks using guards for an infinite number of steps was introduced by Burger et~al.~\cite{ineq1}.
In this paper, we study the concept of \MED, which was introduced by Goddard et~al.~\cite{eternal-security-in-graphs} (eternal domination was originally called eternal security).
Here, the notion of the letter ``m'' emphasizes that multiple guards may move during each round.
There is also a variant of the problem studied by Goddard et~al.~\cite{eternal-security-in-graphs} where only one guard may move during each round, which is not considered in this paper.



The \MED number $\EDN(G)$ is the minimum number of guards which defend against all attacks indefinitely.
Goddard et~al.~\cite{eternal-security-in-graphs} established $\EDN$ exactly for paths, cycles, complete graphs and complete bipartite graphs.
Since then, several results have focused on finding bounds on $\EDN$ under different conditions or graph classes.
Among the studied graph classes are trees~\cite{klostermeyer2021eternal, henning2016trees, Klostermeyer2015}, grids \cite{3n-grids, 5n-grids, Messinger2017, Inerney2021, finbow2020eternal, Inerney2019, Lamprou2019}, and interval graphs~\cite{proper-interval-graphs, Rinemberg2019}.
For a good survey of other related results and topics, see Klostermeyer and Mynhardt~\cite{survey-article}.




Very little is known regarding the algorithmic aspects of \MED.
The decision problem (asking if $\EDN(G)\le k$) is NP-hard and belongs to \EXPTIME, however, it is not known whether it lies in the class \PSPACE~\cite{survey-article}.

\subsection{Original Results}
In this paper, we focus on the class of cactus graphs (connected graphs where each edge lies in at most one cycle) and provide an algorithm for computing $\EDN$ in cactus graphs.
In \Cref{sec:toolbox}, we provide a set of tools with more general applications to proving upper and lower bounds of $\EDN$.
Those tools are then used in \Cref{sec:reducing_cactus_graph} to describe a set of reductions, which allow us to compute $\EDN$ of cactus graphs.
This is a significant expansion of basic principles which were introduced by Klostermeyer and MacGillivray~\cite{eternal-dom-sets}, in which they provide an algorithm for computing $\EDN$ of trees.



Our main result is summarized in the following theorem.
\begin{restatable}{theorem}{cactusalgoritm}\label{thm:cactus-algorithm}
  Let $G$ be a cactus graph on $n$ vertices.
  Then there exists a polynomial algorithm which computes $\EDN(G)$.
\end{restatable}


\subsection{Preliminaries}\label{sec:preliminaries}

Let us now review all the standard concepts formally.
A graph is a \emph{cactus} if its every edge lies on at most one cycle.
For an undirected graph $G$ let a \emph{configuration} be a multiset of its vertices $\CONF=\{c_1,\dots,c_n \mid c_i \in V(G)\}$.
We will refer to the elements of configurations as \emph{guards}.
If a vertex is an element of a configuration, then it is \emph{occupied} (by a guard).
Two configurations $\CONF_1$ and $\CONF_2$ of $G$ are mutually \emph{traversable} if there is some set of pairs $\TRAN(\CONF_1,\CONF_2) = \{(v_1, u_1), (v_2, u_2), \dots, (v_n, u_n)\}$ such that $\CONF_1 = \{v_1,\dots,v_n\}$ and $\CONF_2 = \{u_1,\dots,u_n\}$ and $\{v_i,u_i\} \in E(G)$ for all $i$ from $1$ to $n$.
We perceive the guards as tokens which move through the graph.
The elements of $\TRAN(\CONF_1,\CONF_2)$ are called \emph{movements} and a single ordered pair among them is a \emph{move} of a guard.
A guard that moves in $\TRAN(\CONF_1,\CONF_2)$ to the same vertex where he started is called \emph{stationary}.
A \emph{strategy} in $G$ is a graph $S_G=(\CONFS,\FF)$ where $\CONFS$ is a set of configurations over $V(G)$ such that all of the configurations have the same size and $\FF \subseteq \CONFS^2$ describe possible transitions between the configurations.
The \emph{order} of a strategy is the number of guards in each of its configurations.
In papers on this topic it is often assumed that the strategy edges are given implicitly as $\FF = \big\{\{\CONF_1,\CONF_2\} \in \CONFS^2 \mid \hbox{$\CONF_1$ and $\CONF_2$ are mutually traversable in $G$}\big\}$.
For our purposes, we want to prescribe the strategy explicitly.
We introduce the notions for exact strategy prescription in \Cref{sec:upper_bounds}.

We call the strategy $S_G$ to be \emph{defending against vertex attacks} if for any $\CONF \in \CONFS$ the configuration $\CONF$ and its neighbors in $S_G$ cover all vertices of $G$, i.e., when a vertex $v \in V(G)$ is ``attacked'' one can always respond by changing to a configuration which has a guard at the vertex $v$.
Formally, $S_G = (\CONFS,\FF)$ is defending if
\[
  (\forall C \in \CONFS)\,(\forall v\in V(G))\,\big(v \in C \lor (\exists C' \in \CONFS) (\{C,C'\} \in \FF \land v \in C')\big).
\]

Note that every configuration in a strategy which defends against vertex attacks induces a dominating set in $G$ as otherwise, the attacker would win in the next round.

We investigate two variants of the game.
The variants differ in whether they allow multiple guards to occupy the same vertex.
Let an \emph{m-Eternal Guard Strategy} in $G$ be a strategy defending against vertex attacks in $G$.

\defProblemQuestion{\EDNproblem}
{An undirected graph $G=(V,E)$.}
{What is the minimum number of guards $\EDN$ such that there exists an m-Eternal Guard Strategy $S_G$ where each vertex is occupied by at most one guard that defends against vertex attacks in $G$?}

\defProblemQuestion{\EGCproblem}
{An undirected graph $G=(V,E)$.}
{What is the minimum number of guards $\EGC$ such that there exists an m-Eternal Guard Strategy $S_G$ that defends against vertex attacks in $G$?}

The open neighborhood of $u$ in $G$ will be denoted as $N_G(u)$.
By $P_n$ we denote a path with $n$ edges and $n + 1$ vertices.
By $G[U]$ we denote the subgraph of $G$ induced by the set of vertices $U \subseteq V(G)$.

\section{High-level Overview of the Proof}\label{sec:overview}

In order to solve the \EDNproblem and the \EGCproblem on cactus graphs, we use induction on the number of vertices.
Base cases will be presented in \Cref{def:base_case}.
In the induction step, we show how to reduce cactus graph $G$ to a smaller cactus graph $G'$ while showing lower bound and upper bound in the following ways.
Reduction from $G$ to $G'$ is done using \Cref{obs:identification,obs:lb_leaf,obs:lb_star,obs:lb_path}.
These directly show a lower bound $\EGC(G) \geq \EGC(G') + K$ for some constant $K$.
Then, we show an expansion from $G'$ to $G$.
We assume that $G'$ has an optimal defending strategy that holds several nice properties from the induction.
We show that a part of the graph $G'$ along with its strategy can be exchanged for a different one by showing that \Cref{def:interface_equivalent} holds for them.
Such parts are then exchanged using \Cref{def:expansion} which expands $G'$ into $G$ while showing that an upper bound devised by \Cref{obs:upper_bound} applies.
This gets us an upper bound $\EDN(G) \leq \EDN(G')+K$ (the same $K$ as in the lower bound).
Combining the lower and upper bound using \Cref{lem:technique} gets us the optimal number of guards for $G$.

The used reduction depends on a leaf component that the cactus graph contains by \Cref{obs:blockcut_decomposition}.
One case is that the subgraph is a tree and the second case is that there is a \emph{leaf cycle} -- a cycle with leaves which is connected to the rest of the graph via a single articulation.
We split the reductions into three groups.

The first group called \emph{leaf reductions} shown in \Cref{sec:reducing_trees} has a few simple reductions of leaves which are not incident to a leaf cycle.
These were shown to be sufficient to determine the $\EDN$ for any tree by Goddard, Hedetniemi, and Hedetniemi~\cite{eternal-security-in-graphs}.
We reintroduce these reductions in our framework and show more general results so that the reductions can be used over tree subgraphs of non-tree graph classes.
They also serve as an introductory example of how to use the tools from \Cref{sec:toolbox}.

Further reductions are more involved and require non-trivial manipulation with strategies.
It is beneficial to establish strategies with nice properties in the induction to allow a stronger induction step.
In \Cref{sec:properties}, we show the properties which are used in the two other groups of reductions.

The last two groups called \emph{cycle reductions} and \emph{constant component reductions} are shown in \Cref{sec:cycle_reductions,sec:constant_reductions}.
Cycle reductions concern substructures that appear on leaf cycles.
We fix a leaf cycle and use these reductions repeatedly on it.
Each reduction shortens the leaf cycle.
Eventually, the cycle is very short and is reduced by constant component reductions.
After these reductions, the leaf cycle is removed entirely and only zero, one, or two leaves are left in its place.
Such leaves are then processed either as tree leaves or leaf vertices adjacent to another leaf cycle.

\section{Reducing Trees}\label{sec:reducing_trees}

In this section, we intuitively present tools to achieve lower and upper bounds and which will be formally introduced in \Cref{sec:toolbox}.
We focus on tree reductions, which were first described by Goddard, Hedetniemi, and Hedetniemi~\cite{eternal-security-in-graphs} as a part of the linear algorithm for computing the m-eternal domination number $\EDN$ on trees.
We now show this set of reductions along with the proofs of their correctness in \Cref{lem:reduction-t1,lem:reduction-t2,lem:reduction-t3}.

For graph $G$, let us have a vertex $u \in V(G)$ which is adjacent to $\ell \ge 1$ leaves and has degree $d$.
Let $v$ be one of the leaves adjacent to $u$.
We define three leaf reductions of $G$ to $G'$ as follows.
See~\Cref{tab:leaf_reductions} for an illustration of respective bound proofs.

\begin{reduction}\customlabel{reduction-t1}{$t_1$}
  If $\ell=1$ and $d \le 2$, let $G' = G \setminus \{u,v\}$.
\end{reduction}
\begin{reduction}\customlabel{reduction-t2}{$t_2$}
  If $\ell>2$, let $G' = G \setminus \{v\}$.
\end{reduction}
\begin{reduction}\customlabel{reduction-t3}{$t_3$}
  If $\ell=2$ and $d=3$, let $G' = G \setminus \{\text{all leaves adjacent to $u$}\}$.
\end{reduction}

Reductions~\ref{reduction-t2} and~\ref{reduction-t3} can be joined to a single reduction which removes all leaves of a vertex with $\ell \ge 2$ and $d = \ell +1$ (used in \cite{eternal-security-in-graphs}).
However, Reduction~\ref{reduction-t2} may be used in a wider range of scenarios as it does not require a specific value of $d$.

Assume now, that we know the optimum number of guards for $G'$ (for both $\EGC$ and $\EDN$).
Our goal is to show two things.
By showing that $G$ always uses at least $K$ more guards than $G'$ we get a lower bound on the number of guards necessary for $G$.
By showing that there is a strategy for $G$ which uses at most $K$ more guards than an optimum strategy on $G'$ we get an upper bound on the number of guards on $G$.
Together, these bounds give us an optimum number of guards for $G$.
This concept is formally introduced in \Cref{lem:technique}.

\begin{table}[h]
  \caption[Leaf reductions]{%
    Leaf reductions;
    Lower bound side depicts clique reductions (removal of marked vertices and joining its neighborhood with a clique);
    Upper bound side labels vertices with Greek letters of states where they belong, and arrows show how one state transitions to another.
    The marked groups of vertices are created with \Cref{def:copy}.
  }\label{tab:leaf_reductions}
  \centering
  \begin{tabular}{|c|c|c|}
    \hline
    Reduction & Lower bound & Upper bound \\
    \hline
    \reductioncase{t1}{1}
    \hline
    \reductioncase{t2}{2}
    \hline
    \reductioncase{t3}{3}
    \hline
  \end{tabular}
\end{table}

Having the reductions in hand see \Cref{fig:trees-example} for an example of how the reductions are used to construct a strategy for a tree.
\begin{figure}[h]
  \centering
  \includegraphics{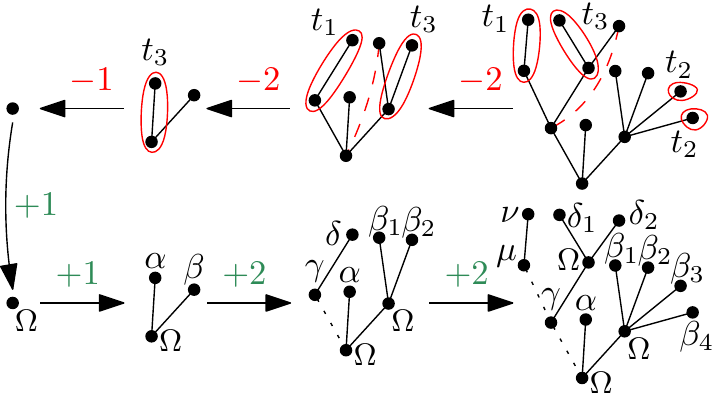}
  \caption[Example of reducing tree graphs]{%
    An example application of leaf reductions on a tree graph.
    Dotted lines signify that the strategy does not use that edge, and the strategies on subtrees are independent, which is caused by Reduction~\ref{reduction-t1}.
    Note how Reduction~\ref{reduction-t3} can be used even when there is no vertex $a$.}%
  \label{fig:trees-example}
\end{figure}

We now proceed to show the bounds obtained from these reductions.
Generally, the proofs contain lower bound and upper bound portions, see \Cref{tab:leaf_reductions} for accompanying illustrations.
Lower bounds can be shown quite easily -- delimit a connected part of a graph which is guaranteed to contain $K$ guards, remove it, and join its neighborhood with a clique.
Upper bounds are more tricky -- we assume some optimal strategy on $G'$, which has nice properties, and then we expand it to $G$ while preserving the properties.
The notation used in the following proofs is defined in \Cref{sec:toolbox}.

We say that graph $G$ is \emph{defended} with $k$ guards if $k = \EDN(G) = \EGC(G)$ and the strategy using $k$ guards is \emph{proper} in the sense of \Cref{prop:labelled_strategy}, which is defined in \Cref{sec:reducing_cactus_graph}.
This allows us the state the lemmas concisely.
Let us now see the proofs for the three leaf reductions.

\begin{lemma}\label{lem:reduction-t1}
  \standardReductionLemma{t1}{1}
\end{lemma}
\begin{proof}
  As $\{u,v\}$ is a leaf and its neighbor, there is always at least one guard so we may apply \Cref{obs:lb_leaf} on $\{u,v\}$ to get a lower bound of $\EGC(G) \geq \EGC(G') + 1$.
  To get an upper bound $\EDN(G) \leq \EDN(G') + 1$, we dedicate one new guard to defend $\{u,v\}$ independently on the rest of the strategy.
  Putting the lower bound and upper bound together using \Cref{lem:technique} we get that $G$ is defended with $1$ more guard than $G'$.
\end{proof}

Note, that the final strategy graph after Reduction~\ref{reduction-t1} is a Cartesian product of the strategy graph on $G'$ and a graph with a single edge.
Cartesian product is a basis for \Cref{def:copy} where we introduce an operation which joins strategies even if the strategies are not entirely independent.
We shall use this operation along with a property shown in \Cref{lem:always_present_guard} -- that a strategy can be altered so that a vertex adjacent to multiple leaves is always occupied.

To ease notation, we shall reserve the prime symbol ($'$) to denote structures of the reduced instance such as the graph $G'$, defending strategy $\mathcal{B}'$, strategy graph $S'_{G'}$, its states (vertices) $\Omega'$ and transitions (edges) $\FF'$, etc.

\begin{lemma}\label{lem:reduction-t2}
  \standardReductionLemma{t2}{0}
\end{lemma}
\begin{proof}
  Lower bound of $0$ is obtained by using \Cref{obs:identification} to identify $v$ with $u$ so $\EGC(G) \geq \EGC(G')$.

  For upper bound, from induction we have a defending \LAB $\mathcal{B}'$ of $G'$.
  We wish to alter it so it defends vertex $v$ as well.
  We apply \Cref{lem:always_present_guard} to alter $\mathcal{B}'$ so that $u$ is occupied in each state of $\Omega'$.
  Let $w$ be a leaf adjacent to $u$ distinct from $v$.
  We partition all states (vertices) of the strategy $S'_{G'}$ as follows.
  A state $\alpha'$ belongs to $\vertexStates'(w)$ if in $\alpha'$ vertex $w$ is occupied.

  Now we perform graph Cartesian product of $S'_{G'}$ with a single edge $\{\alpha,\beta\}$ over subset $\vertexStates'(w)$ (\Cref{def:copy}).
  Written in short as $S_{G'} = \GCPOS{S'_{G'}}{\vertexStates'(w)}{\{\alpha,\beta\}}$.
  This splits all vertices of the strategy where $w$ is occupied into two.
  We denote the new sets as $\alpha$ and $\beta$.
  This got us a new strategy graph $S_{G'}$ over the reduced graph $G'$.
  Now we expand from $G'$ to $G$ while altering the strategy slightly.
  In $\beta$ we substitute the guard on $w$ with a guard on $v$.
  The guards shall transition between states of $\alpha$ and $\beta$ as $\TRAN(\alpha,\beta) = \{(w,u),(u,v)\}$ while the rest of them shall not move.
  As $w$ and $v$ are siblings it follows that we can transition from any $\gamma \in \Omega$ to $w$ the same way as to $v$ if they were swapped.
  This remains defending by \Cref{lem:copy_keeps_domination}.
  Hence, $\EDN(G) \leq \EDN(G')$ and by \Cref{lem:technique} we get that $G$ is defended with the same number of guards as $G'$.
\end{proof}

The shown strategy basically defends $v$ in the ``same way'' it defends $w$.
We can do this when one can transition from one to the other in a single step while the remaining guards remain stationary.
For a detailed explanation see \Cref{def:copy} and its lemmas that show its properties.

The previous reduction bounds were proven with an extensive explanation.
In the following proofs, we just use the tools to arrive at the result directly.
Note that a very similar argument could be used to obtain an arbitrary number of leaves.

\begin{lemma}\label{lem:reduction-t3}
  \standardReductionLemma{t3}{1}
\end{lemma}

\begin{proof}
  Lower bound of $1$ is obtained by using \Cref{obs:lb_leaf} on vertices $\{u,v\}$, which results in a graph isomorphic to one that is created by removing all leaves adjacent to $u$ which gets us $\EGC(G) \geq \EGC(G')+1$.
  For upper bound, we apply \Cref{lem:extend_leaves} which adds the two leaves to $u$ using one extra guard which directly results in $\EDN(G) \leq \EDN(G') + 1$.
  Using \Cref{lem:technique} we get that $G$ is defended with $1$ more guard than $G'$.
\end{proof}


Note that the bounds devised for Reductions~\ref{reduction-t1}, \ref{reduction-t2}, and \ref{reduction-t3} do not require the graph to be a tree.
We may use these reductions in any graph class.
Hence, we may reduce any leaves in subtrees which appear as parts of other graphs.
In particular, Reduction~\ref{reduction-t2} may be also used to reduce the number of leaves adjacent to any vertex to $2$ because connections of $u$ to other vertices do not interfere with the reduction.
Note that in that case, we obtain lower bounds for the \EGCproblem and upper bounds for the \EDNproblem.

It was previously shown by Goddard, Hedetniemi, and Hedetniemi \cite{eternal-security-in-graphs} that these reductions (originally given in a slightly different form) are sufficient to solve any tree graph.
Note that this can be shown by rooting the tree and repeatedly applying Reductions~\ref{reduction-t1}, \ref{reduction-t2}, and~\ref{reduction-t3} on the parent of the deepest leaf.

Also note, that the reductions \ref{reduction-t1} and \ref{reduction-t3} do not require the rest of the graph (signified by vertex $a$) to be there at all, hence, these solve base cases where only a single edge or a star remain.

When the reductions are used on a tree we get a partitioning of vertices into subtrees which are defended independently.
These components constitute a \emph{neo-colonization}, a notion introduced by Goddard et al.~\cite{eternal-security-in-graphs} and often used in contemporary papers.

\section{The \MED Toolbox}\label{sec:toolbox}

This section gives tools to show lower and upper bounds for the \MED problem.
Before we present the approach in detail we show several key ideas and a detailed structure for the rest of this section.
Throughout this paper, we reserve \emph{prime} (e.g. $G'$ and $\alpha'$) to denote structures of the reduced instance.

\begin{observation}\label{obs:edn_lb}
  $\DN(G) \le \EGC(G) \le \EDN(G) \le 2\cdot\DN(G)$ for any graph $G$.
\end{observation}
\begin{proof}
  Every \EDNinstance strategy can be applied as an \EGCinstance strategy so $\EDN \ge \EGC$.
  Every configuration in each of these strategies must induce a dominating set.
  Therefore, they are all lower bound by the domination number $\DN$.

  It is also known that an \EDNinstance strategy can be constructed by defending neighborhood of each vertex in the dominating set independently of each other (with a simple strategy for stars) that uses at most $2\cdot\DN(G)$ guards as shown by Klostermeyer and Mynhardt.~\cite{survey-article}.
\end{proof}

We now show the lemma which sums up how the bounds of the optimal strategies are obtained.

\begin{lemma}\label{lem:technique}
  Let us assume that for graphs $G$, $G'$, and an integer constant $k$
  \begin{align}
    \label{eq:ub} \EDN(G) &\leq \EDN(G') + k,\\
    \label{eq:lb} \EGC(G) &\geq \EGC(G') + k,\\
    \label{eq:ih} \EDN(G') &= \EGC(G').
  \end{align}
  Then $\EDN(G) = \EGC(G) = \EDN(G')+k = \EGC(G')+k $.
\end{lemma}
\begin{proof}
  Given the assumptions, we have
  \[
    \EDN(G) \putabove{\leq}{(\ref{eq:ub})} \EDN(G') + k \putabove{=}{(\ref{eq:ih})} \EGC(G') + k \putabove{\leq}{(\ref{eq:lb})} \EGC(G) \putabove{\leq}{Obs.~\ref{obs:edn_lb}} \EDN(G).
  \]
  As the first and the last term is identical all these values are equal.
\end{proof}

Hence, it suffices to prove that for $G$ and its reduction $G'$ we have $\EDN(G) \leq \EDN(G')+k$ and $\EGC(G') \leq \EGC(G)-k$.
If we already have the optimal strategy for $G'$, then our constructive upper bounds together with \Cref{lem:technique} give us an optimal strategy for $G$.

We present the tools for obtaining lower bounds in \Cref{sec:lower_bounds}, terminology and new concepts for upper bounds in \Cref{sec:upper_bounds}, and tools which use the new concepts to obtain upper bounds in \Cref{sec:upper_bound_tools}.
See \Cref{fig:toolbox} for a detailed section overview.

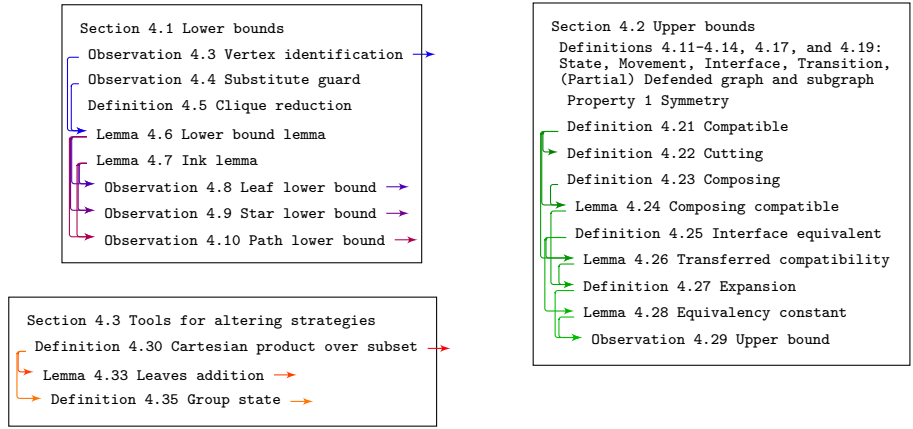
\begin{figure}[h]
    \centering
    \input{images/toolbox.tex}
    \caption[Overview diagram of \Cref{sec:toolbox}]{%
      Overview of \Cref{sec:toolbox}.
      Boxes represent respective subsections;
      Arrows on the left side show which notions are used to prove other notions;
      Right arrows show which notions are frequently used in \Cref{sec:reducing_cactus_graph} to obtain results for cactus graphs.
    }%
    \label{fig:toolbox}
\end{figure}

\subsection{Lower Bounds}\label{sec:lower_bounds}

We start this section with a few elementary observations about strategies.
Then, we show a pair of lemmas which are the main tools in obtaining lower bounds.
Last, using these lemmas, we obtain three lower bound observations which we use frequently in \Cref{sec:reducing_cactus_graph}.

We say that $G'$ is a result of \textit{identifying} $u$ with $v$ in $G$ if it is a result of removing $u$ while adding the edges so that $N_{G'}(v) = N_G(u) \cup N_G(v)$.


\begin{observation}[Vertex identification]\label{obs:identification}
  Let $G$ be a graph and $u$ and $v$ be its two distinct vertices.
  Then for a graph $G'$, which is a result of identifying $u$ with $v$ in $G$, $\EGC(G') \le \EGC(G)$.
\end{observation}
\begin{proof}
  Let $v' \in V(G')$ be the vertex created by identifying $u$ with $v$ in $G$.
  Let $S_G$ be an optimal strategy of $G$.
  Let $S_{G'}$ be a strategy on $G'$ which is the same as $S_G$ except that in every configuration each $u$ and $v$ is substituted by $v'$.

  Any pair of traversable configurations in $S_G$ is still traversable in $S_{G'}$ as in every movement $u$ and $v$ can be replaced by $v'$.
  Any attack on $V(G') \setminus \{v'\}$ is defended by a configuration in $S_{G'}$ which was created from a respective configuration of $S_G$, and $v'$ is defended by a configuration which defended $u$ in $G$.
\end{proof}

\begin{observation}\label{obs:subtitute_guards}
  If a graph has a clique on distinct vertices $u,v,w$, then guard movements $(u,v)$ and $(v,w)$ can be substituted with movement $(u,w)$ and a stationary guard on $v$.
\end{observation}

Along with vertex identification, the following reduction is the main tool for obtaining lower bounds.

\begin{definition}[Clique reduction]\label{def:clique_reduction}
  Let $G$ be a graph and $H$ be its non-empty induced connected subgraph.
  By \emph{clique reduction} of $H$ in $G$ we mean the creation of a new graph $G'$ that is the result of removing $H$ from $G$ and mutually connecting all neighbors of $H$ in $G \setminus H$ by an edge.
\end{definition}

See \Cref{fig:mll} for an illustration of a clique reduction.
Using \Cref{obs:identification,obs:subtitute_guards} we now show that the clique reduction implies a lower bound on $G$ which can be later used to show tight strategy lower bounds.

\begin{lemma}[Lower bound lemma]\label{lem:lower_bound}
  Let $G$ be a graph and $H$ be its non-empty induced connected subgraph such that in at least one optimal m-eternal guard strategy, there are always at least $k$ guards present on $H$.
  Let $G'$ be the result of a clique reduction of $H$ in $G$.
  Then
  \[
    \EGC(G') \le \EGC(G) - k.
  \]
\end{lemma}
\begin{proof}
  If there is no neighbor of $H$ in $G$, then clique reduction removes $H$ and adds no edges.
  We can remove all the guards which were standing on $H$ so $G'$ is clearly defended by $\EGC(G) - k$ guards.

  Otherwise, there is a neighbor of $H$ in $G \setminus H$, say $v$.
  Let $S_G = (\CONFS, \FF)$ be an optimal strategy on $G$.
  We use \Cref{obs:identification} to identify all vertices $V(H)$ with $v$ in $G$ to obtain a subgraph of $G'$ along with a strategy $S_{G'}$.
  Note that in $G'$ each configuration of $S_{G'}$ has at least $k$ guards on $v$ because before identification $H$ always contained $k$ guards.
  Also, configurations which defend $v$ in $S_{G'}$ have at least $k+1$ guards on $v$ by the same argument.

  Let $S^-_{G'}$ be a strategy which is the same as $S_{G'}$, except it has $k$ less guard on $v$ in each configuration.
  We see that each configuration which defended $v$ in $S_G$ had at least $k+1$ guards on $v$ so it defends $v$ in $S^-_{G'}$.
  Guards which defended $V(G) \setminus (V(H) \cup \{v\})$ remain unchanged.
  It remains to check whether configurations which were traversable in $S_G$ remain traversable in $S^-_{G'}$.

  Our goal is to show that there exists a set of movements between each pair of configurations of $S_{G'}$ which have $k$ stationary guards on $v$ which are not needed for defending $G'$.
  Such guards then may be removed to obtain $S^-_{G'}$ and the remaining movements show that the respective configurations are traversable.

  Each movement $(u,h)$ and $(h,w)$ such that $h \in V(H)$ in $S_G$ has its respective pair of movements $(u,v)$, $(v,w)$ in $S_{G'}$.
  As $u$ and $w$ are neighbors of $V(H)$ then there is an edge $\{u,w\} \in E(G')$ added by the construction of $G'$.
  By \Cref{obs:subtitute_guards} we may substitute movements $(u,v)$, $(v,w)$ with $(u,w)$ and a stationary guard on $v$.

  Assume there are less than $k$ stationary guards on $v$ in $S_{G'}$ after applying the substitution exhaustively.
  Then there must be at most $k-1$ stationary guards and at least one guard which leaves $v$ or at least one guard which arrives to $v$, but there may not be both (one leaving and one arriving) as they would form $(u,v)$, $(v,w)$ pair and the substitution could be applied.
  When the guard is leaving or arriving there are at most $k-1$ guards in the final or starting configuration, respectively, which is a contradiction because there are at least $k$ guards on $v$.

  Removing $k$ guards from $v$ in $S_{G'}$ yields $S^-_{G'}$ where each configuration pair remains traversable, which concludes the proof.
\end{proof}

To use \Cref{lem:lower_bound} we need to show that an induced subgraph $H$ of $G$ is always occupied by at least $k$ guards.
To do that we have the following lemma.

\begin{lemma}[Ink lemma]\label{lem:ink}
  Let $H$ be an induced subgraph of $G$.
  Let $(v_1, v_2, \dots, v_k)$ be a sequence of vertices in $H$ such that it holds $d(u, v_i) > i$ for every $i$ and for every $u \in V(G) \setminus V(H)$, and also for every $j < i$ it holds that $d(v_j, v_i) > i - j$.
  Then there are at least $k$ guards on $H$ in every defending m-eternal guard configuration.
\end{lemma}
\begin{proof}
  Let $\CONF$ be any fixed configuration of a defending strategy.
  We show that $\CONF$ contains at least $k$ guards on $H$.

  Assume that the attacker performed a sequence of attacks $(v_1, v_2, \dots, v_k)$ one by one.
  At the $i$-th step of the attack sequence, the following is true.
  Guards who were standing on $V(G) \setminus V(H)$ at the beginning of the attack sequence are more than $i$ edges far from $v_i$ so they cannot reach $v_i$ in time to defend it.
  Similarly, any guard that defended $v_j$ with $j < i$ can not defend the attack on $v_i$ as their distance from $v_i$ is more than $i - j$ at the time they defended $v_j$.
  Therefore, none of the guards can reach $v_i$ in time and we need an additional guard placed on $H$.

  In total, we need $k$ guards on $H$ in a configuration to be able to defend the attack sequence.
  This is true for any configuration so every defending strategy must have $k$ guards on $H$ in every configuration.
\end{proof}

The operation in \Cref{lem:lower_bound} together with the lower bound obtained from \Cref{lem:ink} allows us to make a graph smaller while showing that the removed part required some minimum number of guards.
See an example usage of \Cref{lem:lower_bound,lem:ink} in \Cref{fig:mll}.

\begin{figure}[h]
  \centering
  \includegraphics{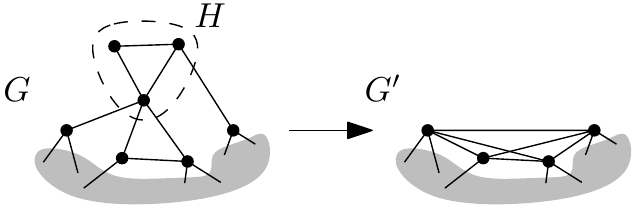}
  \caption[Example application of \Cref{def:clique_reduction}]{%
    A graph $G$ with an induced subgraph $H$.
    Graph $G'$ is obtained by a clique reduction from \Cref{def:clique_reduction}.
    By \Cref{lem:ink} we have that every configuration contains at least $1$ guard on $H$.
    Hence, by \Cref{lem:lower_bound} we have $\EGC(G) \ge \EGC(G') + 1$.
  }%
  \label{fig:mll}
\end{figure}

\begin{observation}\label{obs:lb_leaf}
  In graph $G$, let $v$ be a leaf vertex and let $u$ be its neighbor, and let $H = G[\{u,v\}]$.
  By \Cref{lem:ink} with a sequence $(v)$ we obtain that $1$ guard is on $H$.
  In other words, the closed neighborhood of every leaf must contain at least one guard otherwise an attack on the leaf could not be defended.
  Let $G'$ be a graph obtained by using the operation of \Cref{lem:lower_bound} on $H$, this gives us $\EGC(G) \ge \EGC(G') + 1$.
\end{observation}

\begin{observation}\label{obs:lb_star}
  In graph $G$, let $u$ be a vertex which is adjacent to at least two leaves $\{v_1,v_2,\dots\}$, and let $H=N[u]$ denote the closed neighborhood of $u$.
  By \Cref{lem:ink} with a sequence $(v_1,v_2)$ we obtain that $2$ guards are on $H$.
  In other words, the closed neighborhood of $u$ must contain at least $2$ guards otherwise two consecutive attacks on different leaves adjacent to $u$ could not be defended.
  Let $G'$ be a graph obtained by using the operation of \Cref{lem:lower_bound} on $H$, this gives us $\EGC(G) \ge \EGC(G') + 2$.
\end{observation}

\begin{observation}\label{obs:lb_path}
  Let us have graph $G$ and its induced subgraph $H$ that is isomorphic to a path on three vertices.
  We label these three vertices of $H$ as $u_1,u_2,u_3$ (in order).
  By \Cref{lem:ink} with a sequence $(u_2)$ we have the lower bound of $1$ on the number of guards on $H$.
  Let $G'$ be a graph obtained by using the operation of \Cref{lem:lower_bound} on $\{u_1,u_2,u_3\}$.
  This gives us $\EGC(G) \ge \EGC(G') + 1$.
\end{observation}

\subsection{Upper Bounds}\label{sec:upper_bounds}

This section introduces notation to describe strategies which are used to achieve upper bounds for $\EDN$.
We assume that we have a graph $G$ and its reduced copy $G'$.
The main idea is that a strategy for $G'$ can be locally changed to obtain a strategy for $G$.
To accommodate this local change, we show how to cut and compose parts of the graph while preserving its strategy.
At the end of this section, we present a set of sufficient rules that allow such a local change.
Then, in \Cref{sec:upper_bound_tools}, we present tools which we use to obtain upper bounds.

In our constructions, we need to have control over the movements of the guards.
We also need a way to represent only part of the strategy over an induced subgraph of $G$.
To do so, we introduce states (labelled configurations) and \LAB that prescribes the guard movements on state transition.

\begin{definition}[States]\label{def:states}
  Let \emph{states} be a set of labels $\Omega$ and let \emph{state vertex mapping $P$ of $\Omega$ to $V(G)$} be $P \colon \Omega \to 2^{V(G)}$, i.e., a state $\alpha \in \Omega$ represents a subset of vertices $P(\alpha) \subseteq V(G)$ (also called guards) of a graph $G$.
  Let $\vertexStates(v) = \{ \beta \mid \beta \in \Omega, v \in P(\beta) \}$ (states that contain $v$) for every $v \in V(G)$.
\end{definition}

We will use Greek letters such as $\alpha,\beta,\gamma,\delta,\varphi$ to signify states or sets of states.
\emph{Move} of a guard is still an ordered pair of vertices $(u,v)$ such that $\{u,v\} \in E(G)$ or $u=v$ (\emph{stationary} guard).

Building towards a comprehensive definition of a \LAB, we first build a more general concept -- partial \LAB.
This will allow us to do cutting and composing with a well-defined strategy over a subgraph.

\begin{definition}[Interface]\label{def:interface}
  Let an \emph{interface} $R$ of a graph $G$ with respect to its supergraph $H$ be a subset of vertices such that
  \[
    R = \big\{u \mid (\exists v)\, u,v\in V(H), \{u,v\}\in E(H), u\in V(G), v\not\in V(G)\big\},
  \]
  i.e., those vertices of $G$ which have a neighbor in $V(H) \setminus V(G)$ in $H$ .
\end{definition}
\begin{figure}[h]
  \centering
  \includegraphics{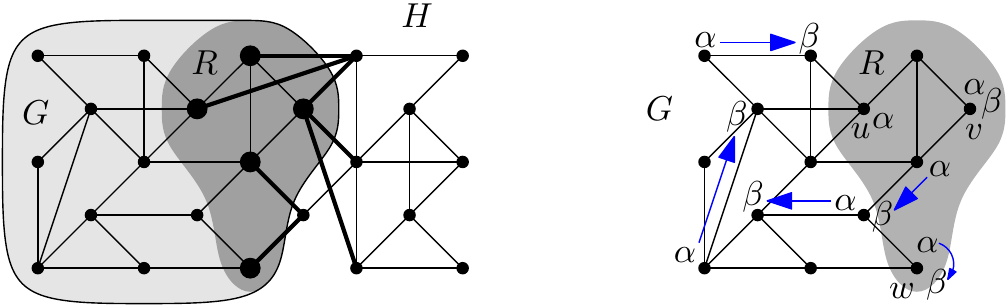}
  \caption[Graph interface and transitions]{%
    Left: An interface $R$ of $G$ with respect to $H$.
    Bold edges signify the cut between $V(G)$ and $V(H)\setminus V(G)$ that is responsible for the vertices in $R$;
    Right: Transition $\TRAN(\alpha,\beta)$ from $\alpha$ to $\beta$.
    Arrows signify movements of the transition.
    All vertices in states $\alpha$ and $\beta$ must be paired up with a movement if they are not in the interface $R$.
    In the interface, a guard moves from $u$ to the rest of the graph outside of $G$ so a movement is missing for $u$.
    Note the difference between $v$ and $w$: in $w$ the same guard stays on the vertex, in $v$ (as there is no $(v,v)$ movement) the guard on $v$ moves out and a different one moves to $v$.
  }%
  \label{fig:example_interface}
\end{figure}

See \Cref{fig:example_interface} for an example of an interface.
The interface marks the vertices where the strategy may be incomplete.
The transitions between states incorporate the interface by allowing the moves to be incomplete in the following way.

\begin{definition}[Transition]\label{def:transition}
  For states $\alpha$ and $\beta$ and a graph $G$ with an interface $R$, let a \emph{transition} (from $\alpha$ to $\beta$) denoted by $\TRAN(\alpha,\beta)$ be a set of moves such that
  \begin{itemize}
    \item $\TRAN(\alpha,\beta) \subseteq \big\{(u,v) \mid u\in P(\alpha), v\in P(\beta), \{u,v\} \in E(G) \vee u = v\big\}$,
    \item for each $u \in P(\alpha)\setminus R$ there exists exactly one $(u,v) \in \TRAN(\alpha,\beta)$,
    \item for each $v \in P(\beta)\setminus R$ there exists exactly one $(u,v) \in \TRAN(\alpha,\beta)$,
    \item for each $u \in P(\alpha)\cap R$ there exists at most one $(u,v) \in \TRAN(\alpha,\beta)$, and
    \item for each $v \in P(\beta)\cap R$ there exists at most one $(u,v) \in \TRAN(\alpha,\beta)$.
  \end{itemize}
\end{definition}

See \Cref{fig:example_interface} for an example of a transition and how it interacts with an interface.
Note that if the interface $R$ is empty, then the transition yields a bijection between guards of the states, which gives an exact prescription on how they move between the two states.

Transition gives us that each guard can be in relation with at most one other guard.
In our case, there is at most one guard on each vertex.
Hence, we may use the standard relation terminology for the set of pairs defined by a transition.

\begin{definition}[Partial \LAB]\label{def:partial}
  A \emph{partial \LAB} is $(G, S_G, P, \TRAN, R)$ where $G$ is a graph, $S_G=(\Omega,\FF)$ is a \emph{strategy graph} such that $\Omega$ is a set of vertices (states) and $\FF$ is a set of edges, $P$ is a state vertex mapping of $\Omega$ to $V(G)$, $R \subseteq V(G)$ is an interface of $G$, and $\TRAN$ maps orientations $(\alpha,\beta)$ and $(\beta,\alpha)$ of every edge $\{\alpha,\beta\} \in \FF$ to transitions $\TRAN(\alpha,\beta)$ and $\TRAN(\beta,\alpha)$, respectively.
\end{definition}
For various purposes, it may be beneficial to think of the strategy graph $S_G$ as an oriented graph (allowing non-symmetric transitions, see \Cref{prop:strategy_symmetry}), or even multigraph (allowing multiple different transitions between the same set of states).

The \LAB may defend against attacks indefinitely if it is in accordance with the following definition.
\begin{definition}[Defending]\label{def:defending}
  A partial \LAB $(G, (\Omega,\FF), P, \TRAN, R)$ is \emph{defending} $G$ if for every state $\alpha \in \Omega$ and each vertex $v\in V(G)$ there is $\beta \in \Omega$ such that $\{\alpha,\beta\} \in \FF$ and $v\in P(\beta)$.
\end{definition}
In other words, \Cref{def:defending} says that for each vertex of the graph every state is either occupying it or a state which occupies it is reachable with only one transition.
This directly leads to the following observation.
Recall that $\vertexStates(u)$ denotes a set of states that contain $u$.

\begin{observation}\label{obs:defending_dominating_set}
  A strategy is defending a graph if for every $u \in V(G)$ set $\vertexStates(u)$ is a dominating set of the strategy graph $S_G$.
\end{observation}


\begin{definition}[Labelled strategy]\label{def:defended_graph}
  A \emph{\LAB} is $\DG = (G, S_G, P, \TRAN)$ such that $(G, S_G, P, \TRAN, \emptyset)$ is a partial \LAB.
\end{definition}

Note, that all the states in the \LAB must contain the same number of guards because the transitions are bijections.
When the strategy is optimal the number of guards corresponds to $\EDN$.

Partial \LABS can have several nice properties, which we present now.
When the strategy graph is unoriented it is natural to require symmetry of transitions.

\begin{property}[Symmetry]\label{prop:strategy_symmetry}
  A partial \LAB $\mathcal{B} = (G, (\Omega,\FF), P, \TRAN, R)$ is \emph{symmetrical} if and only if $\TRAN(\alpha,\beta)$ is a converse relation to $\TRAN(\beta,\alpha)$ (i.e., $\TRAN(\alpha,\beta) = \{(a,b) \mid (b,a) \in \TRAN(\beta,\alpha)\}$) for every $\{\alpha,\beta\} \in \FF$.
\end{property}
\begin{lemma}\label{lem:symmetry}
  For each partial \LAB $\mathcal{B} = (G, (\Omega,\FF), P, \TRAN, R)$ there exists a symmetrical partial \LAB $\mathcal{B}' = (G, (\Omega,\FF), P, \TRAN', R)$.
\end{lemma}
\begin{proof}
  For each pair of states $\{\alpha,\beta\} \in \FF$ fix an arbitrary orientation $(\alpha,\beta)$ and take the $\TRAN(\alpha,\beta)$ with an interface $R$ which gives us $\TRAN(\alpha,\beta)$.
  Note that by swapping $u$ with $v$ and $\alpha$ with $\beta$ in \Cref{def:transition} we obtain the same definition but for $\TRAN(\beta,\alpha)$.
  We substitute the transition $\TRAN(\beta,\alpha)$ for this newly found transition.
  Performing this substitution for every pair of states in $\FF$ gives us the desired $\TRAN'$.
\end{proof}
By \Cref{lem:symmetry}, we will always assume that the partial \LAB is symmetrical.
We also use this property to infer transitions.
We show only one direction of the transition mapping and let the other direction be the converse transition given by symmetry.

It is not easy to grasp the \LAB description only from the formal notation so we shall draw many auxiliary pictures.
Vertices and edges shall be depicted by small circles (or squares) and line segments, respectively;
vertices may be labelled by their letter name;
by Greek letters we signify the states which contain respective vertices;
guard moves in transitions are depicted by differently styled arrows on edges which point between the state labels
(Note that the arrows are always shown only in one direction because we assume \Cref{prop:strategy_symmetry}.);
the interface vertices are marked by gray-filled areas;
see \Cref{fig:example_defended_graph} for an example of a \LAB.

\begin{figure}[h]
  \centering
  \includegraphics{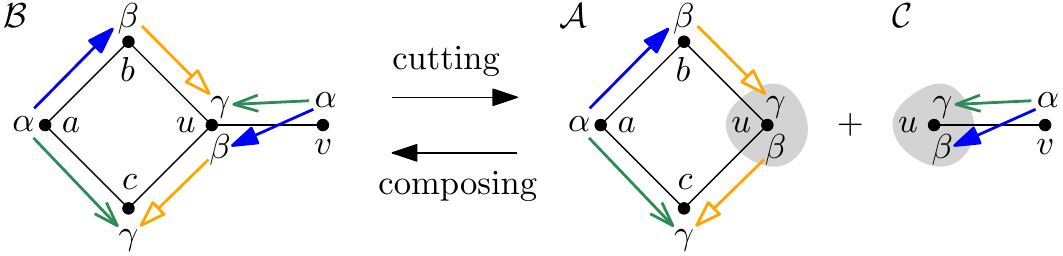}
  \caption[Example (partial) \LAB, cutting, and composing]{%
    Example of a \LAB $\mathcal{B}$ and two partial \LABS $\mathcal{A}$ and $\mathcal{C}$;
    the full formal description of the \LAB $\mathcal{B} = (G,(\Omega,\FF),P,\TRAN)$ is $G=(V,E)$,
    $V=\{a,b,c,u,v\}$,
    $E=\big\{\{a,b\},\{b,u\},\{u,v\},\{a,c\},\{c,u\}\big\}$,
    $\Omega=\{\alpha,\beta,\gamma\}$,
    $\FF=\big\{\{\alpha,\beta\},\{\alpha,\gamma\},\{\beta,\gamma\}\big\}$,
    $P(\alpha)=\{v,a\}$,
    $P(\beta)=\{u,b\}$,
    $P(\gamma)=\{u,c\}$,
    ($R=\emptyset$),
    $\TRAN(\alpha,\beta)=\{(a,b),(v,u)\}$,
    $\TRAN(\alpha,\gamma)=\{(a,c),(v,u)\}$,
    $\TRAN(\beta,\gamma)=\{(b,u),(u,c)\}$;
    the formal descriptions of partial \LABS $\mathcal{A}$ and $\mathcal{C}$ are similar while restricted to their subgraph and they contain an interface $R=\{u\}$.}%
  \label{fig:example_defended_graph}
\end{figure}

We will need to cut part of the \LAB and put something slightly different in its place.
To tackle that we put forward the following notions.

\begin{definition}[Partial labelled substrategy]\label{def:partial_subgraph}
  The \emph{partial labelled substrategy} $\mathcal{B'}$ of a \LAB $\mathcal{B}=(G,(\Omega,\FF),P,\TRAN)$ for some induced subgraph $G'$ of $G$ is a partial \LAB $\mathcal{B}' = (G',(\Omega,\FF),P',\TRAN',R)$ where $R$ is an interface of $G'$ with respect to $G$, for all $\alpha \in \Omega$ it holds $P'(\alpha) = P(\alpha) \cap V(G')$, and for all $\beta, \gamma \in \Omega$ it holds $\TRAN'(\beta, \gamma) = \{ (a, b) \mid (a,b) \in \TRAN(\beta, \gamma) \land a, b \in V(G')\}$.
\end{definition}

It is not immediately obvious that the \Cref{def:partial_subgraph} made a partial labelled substrategy in a way that it constitutes a partial \LAB; so we show that next.

\begin{observation}\label{obs:subgraph_is_partial_graph}
  A partial labelled substrategy $\mathcal{B'} = (G',(\Omega,\FF),P',\TRAN',R)$ of a \LAB $\mathcal{B}=(G,(\Omega,\FF),P,\TRAN)$ is a partial \LAB.
\end{observation}
\begin{proof}
  $G'$ is a graph, $\Omega$ is a set of labels, and $R$ is a subset of $V(G')$ which is in accordance to \Cref{def:partial}.
  $P'$ was created by restricting $P$ to the vertices of $V(G')$.
  We only removed some guards from the mapping so this is okay by \Cref{def:partial}.
  Last, the only guards which are not included in $\TRAN'$ are those whose moves in $\TRAN$ went outside of $G'$.
  Assume such guard on vertex $u$ with a move $(u,v)$ where $v \not\in G'$, hence, $u \in R$ by \Cref{def:interface}.
  As stated in \Cref{def:transition} any guard in $R$ does not have to be included in a move so any partial labelled substrategy is a partial \LAB.
\end{proof}

Now we present conditions which are necessary to be able to combine two partial \LABS into one \LAB.
Based on that, we show how to split a \LAB into two partial \LABS.
\Cref{fig:example_defended_graph} shows an example of the following operations.

\begin{definition}[Compatible]\label{def:compatible}
  Two partial \LABS $\mathcal{B}_1$ and $\mathcal{B}_2$ (denoted as $\mathcal{B}_i = (G_i, (\Omega_i,\FF_i), P_i, M_i, R_i)$) are called \emph{compatible} if the following conditions hold true.
  \begin{itemize}
    \item $R_1 = R_2 = V(G_1) \cap V(G_2)$, i.e., their graphs overlap exactly in the interface,
    \item $(\Omega_1,\FF_1) = (\Omega_2,\FF_2)$, i.e., the strategy graphs are the same,
    \item $M_1(\alpha,\beta) \cup M_2(\alpha,\beta)$ is a bijection between $P_1(\alpha) \cup P_2(\alpha)$ and $P_1(\beta)  \cup P_2(\beta)$ for every $\alpha, \beta \in \Omega_1$.
  \end{itemize}
\end{definition}

The conditions for compatible partial \LABS ensure that the interfaces overlap in a way that a composed function will be a bijection which allows us to cut and compose them in the following way.

\begin{definition}[Cut]\label{def:cut}
  Let us have a \LAB $\mathcal{B}$ and a vertex cut $R$ which partitions the vertices of $G(\mathcal{B})$ into $R$, $A$ and $C$ in such a way that there are no edges between $A$ and $C$.
  We say $\mathcal{B}$ is \emph{cut} along $R$ into two partial labelled substrategies $\mathcal{A}$ and $\mathcal{C}$ where $\mathcal{A}$ is a partial labelled substrategy induced by $V(G(\mathcal{A})) = R \cup A$ and $\mathcal{C}$ is a partial labelled substrategy induced by $V(G(\mathcal{C})) = R \cup C$ such that $\mathcal{A}$ and $\mathcal{C}$ are compatible.
\end{definition}


\begin{definition}[Composing]\label{def:composing}
  By \emph{composing} two partial \LABS $\mathcal{B}_1$ and $\mathcal{B}_2$ ($\mathcal{B}_i = (G_i, (\Omega_i,\FF_i), P_i, M_i, R_i)$) we mean getting $(G^*, (\Omega,\FF), P^*, \TRAN^*)$ where $G^* = \big(V(G_1) \cup V(G_2), E(G_1) \cup E(G_2)\big)$, $\Omega = \Omega_1 \cup \Omega_2$, $\FF = \FF_1 \cup \FF_2$, $\forall \gamma\in \Omega$ we have $P^*(\gamma) = P_1(\gamma) \cup P_2(\gamma)$, and $\TRAN^*(\alpha,\beta) = M_1(\alpha,\beta)\cup M_2(\alpha,\beta)$ for every $\alpha,\beta \in \Omega$.
\end{definition}

\begin{lemma}\label{lem:compatible_composition}
  Composing two compatible partial \LABS yields a \LAB.
\end{lemma}
\begin{proof}
  Let us use the notation of \Cref{def:compatible,def:composing}.
  We need to check whether $(G^*, (\Omega,\FF), P^*, \TRAN^*, \emptyset)$ is a partial \LAB.
  First, $G^*$ is a graph where we unite vertices and edges, while only the interface vertices are overlapping; this constitutes a well-defined graph without multiedges and loops.
  The states $\Omega_1$ are the same for the compatible $\mathcal{B}_1$ and $\mathcal{B}_2$.
  Next, mapping of the states to vertices is done by uniting the individual sets $P_1(\gamma) \cup P_2(\gamma)$ for each $\gamma \in \Omega_1$.
  Each vertex is now guarded in the union of states it was guarded before.
  Last, we check whether the union of $M_1$ and $M_2$ always maps to a well-defined transition, however, this is ensured by compatibility conditions over $P_i$ and $M_i$ in \Cref{def:compatible}.
\end{proof}

We define an equivalency relation (reflexive, symmetric, and transitive) with respect to the interfaces as follows.
\begin{definition}[Interface equivalent]\label{def:interface_equivalent}
  Two partial \LABS $\mathcal{B}_1$ and $\mathcal{B}_2$ ($\mathcal{B}_i = (G_i, (\Omega_i,\FF_i), P_i, M_i, R_i)$) are \emph{interface equivalent} if\/ $G[R_1] = G[R_2]$, $\Omega_1 = \Omega_2$, $\FF_1 = \FF_2$, for all $\alpha \in \Omega_1$ we have $P_1(\alpha) \cap R_1 = P_2(\alpha) \cap R_2$, and we have $(a,b) \in M_1(\beta,\gamma) \Leftrightarrow (a,b) \in M_2(\beta,\gamma)$ for all $u$ such that $a=u \lor b=u$ for all $\beta,\gamma \in \Omega_1$.
\end{definition}

Interface equivalent partial \LABS have the same states with respect to the interface.
This allows us to infer compatibility as stated in the following lemma.

\begin{lemma}\label{lem:transfered_compatibility}
  For three partial \LABS $\mathcal{B}_1$, $\mathcal{B}_2$, and $\mathcal{B}_3$ if $\mathcal{B}_1$ is compatible with $\mathcal{B}_2$, $\mathcal{B}_2$ is interface equivalent with $\mathcal{B}_3$, and $V(G(\mathcal{B}_1)) \cap V(G(\mathcal{B}_3)) = R(\mathcal{B}_3)$, then $\mathcal{B}_1$ is compatible with $\mathcal{B}_3$.
\end{lemma}
\begin{proof}
  Let $\mathcal{B}_i = (G_i, (\Omega_i,\FF_i), P_i, M_i, R_i)$.
  We will check the conditions stated in \Cref{def:compatible}.
  As $V(G_1) \cap V(G_3) = R_3$ and $R_3 = R_2$ by interface equivalency, and $R_2 = R_1$ by compatibility, the first condition holds.
  As $G_3[R_3] = G_2[R_2]$ and $V(G_1) \cap V(G_3) = R_3$ there are no possible edges which would be shared by $G_1$ and $G_3$ outside of $R_3$.
  $\Omega_1 = \Omega_2$ by their compatibility, $\Omega_2 = \Omega_3$ by interface equivalency, so $\Omega_1 = \Omega_3$.
  The $\mathcal{B}_3$ is a partial \LAB so each guard on a vertex in $V(G_3) \setminus R_3$ is covered by $M_3$ exactly once.
  The guards on $R_3$ are covered exactly when they were covered on $R_2$.
  As $\mathcal{B}_1$ and $\mathcal{B}_2$ are compatible the guards on $R_2$ were covered by $M_1$ exactly when they were not covered by $M_2$ and vice-versa.
  Hence, this property still holds for $\mathcal{B}_1$ and $\mathcal{B}_3$.
\end{proof}
The culmination of the previous notions and lemmas is the following procedure which we use as one major part for proving upper bounds.

\begin{definition}[Expansion]\label{def:expansion}
  Let us have a \LAB $\mathcal{B}$ with a partial labelled substrategy $\mathcal{C}$.
  Let us also have a partial \LAB $\mathcal{C}'$ which is interface equivalent with $\mathcal{C}$.
  An \emph{expansion} of $\mathcal{B}$ from $\mathcal{C}$ to $\mathcal{C}'$ is the following sequence of operations.
  \begin{itemize}
    \item Cutting $\mathcal{B}$ along $R(\mathcal{C})$ into $\mathcal{C}$ and $\mathcal{D}$ (see \Cref{def:cut}),
    \item composing $\mathcal{D}$ with $\mathcal{C}'$ into a \LAB $\mathcal{R}$ (see \Cref{def:composing}).
  \end{itemize}
  Partial \LABS $\mathcal{D}$ and $\mathcal{C}'$ are compatible due to \Cref{lem:transfered_compatibility}.
  The result $\mathcal{R}$ is a \LAB due to \Cref{lem:compatible_composition}.
\end{definition}

To establish the difference in the number of guards used to defend $\mathcal{B}$ and $\mathcal{R}$ we have the following lemma.

\begin{lemma}\label{lem:equivalency_constant}
  For two interface equivalent partial \LAB $\mathcal{B}_1$ and $\mathcal{B}_2$ (as in \Cref{def:interface_equivalent}) there is some constant $K(P_1,P_2) \in \mathbb{Z}$ such that for all $\alpha \in \Omega_1$ we have
  \[
    K(P_1,P_2) = |P_2(\alpha)| - |P_1(\alpha)|.
  \]
\end{lemma}
\begin{proof}
  Suppose we have arbitrary states $\alpha, \beta \in \Omega_1$ and let $K_\alpha = |P_2(\alpha)| - |P_1(\alpha)|$ and $K_\beta = |P_2(\beta)| - |P_1(\beta)|$.
  Let $M_1(\alpha,\beta)$ be part of $\mathcal{B}_1$ and $M_2(\alpha,\beta)$ part of $\mathcal{B}_2$.
  Each defines a pairing of guards in respective states.
  However, the guards on the interface are not required to participate in the pairing.
  So we have $|P_1(\alpha)| + g_1(\alpha,\beta) = |P_1(\beta)| + g_1(\beta,\alpha)$ where $g_i$ is the number of guards that do not participate in the pairing of respective $M_i$ (we assume symmetric moves).
  Similarly for $\mathcal{B}_2$ we have $|P_2(\alpha)| + g_2(\alpha,\beta) = |P_2(\beta)| + g_2(\beta,\alpha)$.

  As the partial \LABS are interface equivalent, the sets of guards which do not participate in the pairings is the same, so $g_1(\gamma,\delta) = g_2(\gamma,\delta)$ for all $\gamma,\delta \in \Omega_1$.
  We get
  \begin{align*}
    K_\alpha &= |P_2(\alpha)| - |P_1(\alpha)|\\
             &= |P_2(\beta)| + g_1(\beta,\alpha) - g_1(\alpha,\beta) - (|P_1(\beta)| + g_2(\beta,\alpha) - g_2(\alpha,\beta))\\
             &= |P_2(\beta)| - |P_1(\beta)| + (g_1(\beta,\alpha) - g_2(\beta,\alpha)) + (g_2(\alpha,\beta) - g_1(\alpha,\beta))\\
             &= |P_2(\beta)| - |P_1(\beta)| = K_\beta.
  \end{align*}
  We set $K(P_1,P_2) = K_\alpha$ as we showed that this value is the same irrespective of the chosen $\alpha$.
\end{proof}

To be able to use an expansion we need to select a partial labelled substrategy $\mathcal{C}$ of $\mathcal{B}$ and then show that $\mathcal{C}$ is interface equivalent with $\mathcal{C}'$.
The expansion then proceeds as in \Cref{def:expansion} and an upper bound is obtained from the following observation.

\begin{observation}\label{obs:upper_bound}
  Let us have an expansion of $\mathcal{B}$ from $\mathcal{C}$ to $\mathcal{C}'$ (\Cref{def:expansion}) which results in a \LAB $\mathcal{R}$.
  The expansion increases the number of used guards by $K(P(\mathcal{C}),P(\mathcal{C}'))$ due to \Cref{lem:equivalency_constant}.
  Assuming that $\mathcal{B}$ is an optimal strategy we obtain $\EDN(G(\mathcal{R})) \le \EDN(G(\mathcal{B})) + K(P(\mathcal{C}),P(\mathcal{C}'))$.
\end{observation}

We showed a way to describe a \LAB and how we can exchange the underlying defended graph.
However, to be able to do this we need the strategy to be the same for the original and expanded graph.
So before we start expansion we alter the strategy on the original graph.
This is discussed in the following section.

\subsection{Tools for Altering Strategies}\label{sec:upper_bound_tools}

In this section, we introduce further notions useful for working with strategies when building upper bound constructions.
The typical upper bound proof uses tools introduced in this section to alter the strategy and then applies expansion (\Cref{def:expansion}) which gives the upper bound by \Cref{obs:upper_bound}.

First, let us note that all the notions can be thought of as ``up to isomorphism'' because we can relabel graph or strategy vertices and relabeling does not fundamentally change them.
We skipped this in definitions for the sake of readability.
Let us also set from now on $\mathcal{B} = (G, S_G, P, \TRAN, R)$ and $S_G = (\Omega,\FF)$, and similarly for $\mathcal{B}'$ and $\mathcal{B}^*$ have respective graphs $G'$ and $G^*$, strategies, mappings, etc.

Now we present the main operation for altering strategy graphs.

\begin{definition}[Graph Cartesian product over subset]\label{def:copy}
  Let us have graphs $G_1$ and $G_2$ while $A \subseteq V(G_1)$.
  The \emph{graph Cartesian product over subset $A$} denoted as $\GCPOS{G_1}{A}{G_2}$ is a graph $H$ such that
  \begin{align*}
    V(H) = &\,\{(u,\emptyset) \mid u \in V(G_1) \setminus A\} \cup \{(u,v) \mid u \in A, v \in V(G_2)\}, \\
    \{(a,b),(c,d)\} \in E(H) \Leftrightarrow &\,\big((a = c) \land (a \in A) \land (\{b,d\} \in E(G_2))\big) \lor  \\
    \lor &\,\big(\{a,c\} \in E(G_1) \land ((b = d) \lor (b = \emptyset) \lor (d = \emptyset))\big).
  \end{align*}
\end{definition}

The operation in \Cref{def:copy} can be thought of as a Cartesian product where the sets of vertices created from $G_1$ which are not present in $A$ are identified to a single vertex.
Equivalently, $H$ can be constructed by taking the graph Cartesian product of $G[A]$ and $H$, adding $G[V(G) \setminus A]$, relabeling each new vertex $u$ as $(u, \emptyset)$ and connecting each such $(u, \emptyset) \in V(G) \setminus A$ to all $(v, x) \in A \times V(H)$ such that $v \in N_{G_1}(u)$.
This operations will prove very useful when altering strategy graphs -- we will see it used soon in \Cref{lem:always_present_guard} and many times in \Cref{sec:reducing_cactus_graph}.
The aim of this operation is to defend parts of the graph almost independently.
The edges created from $G_1$ represent changes of guard positions within one part of the graph and edges from $G_2$ represent changes in another part.
While guards move within one part of the graph then the guards in the other part will remain stationary.
The necessity of the set $A$ comes from the fact that the strategy in one part assumes that a guard occupies vertex (e.g. $u$) so then the altered part is restricted to vertices where the guard is present on the vertex ($A = \vertexStates(u)$).
See \Cref{fig:cart_product} for an example application of the Cartesian product over subset.

\begin{figure}[h]
  \centering
  \includegraphics[scale=1.3]{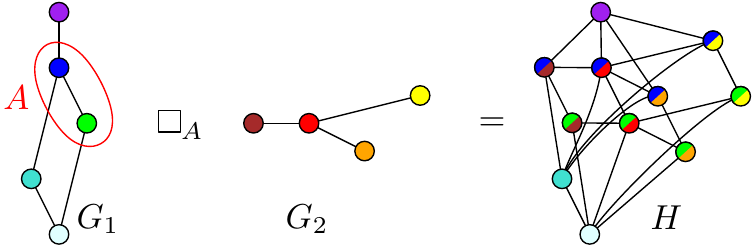}
  \caption[Cartesian product over subset]{Example of a graph Cartesian product of $G_1$ and $G_2$ over subset $A$.}%
  \label{fig:cart_product}
\end{figure}

We shall use the Cartesian product of $G'$ and complete graph over subset very often so we will use short notation that allows us to focus on what happens in the created strategy.

\begin{definition}[Short notation]\label{def:short_copy_notation}
  Let $\GCPOS{G}{A}{\{\alpha_1,\alpha_2,\dots,\alpha_n\}} = \GCPOS{G}{A}{K_n}$ where $V(K_n) = \{\beta_1,\dots,\beta_n\}$ and $\alpha_i$ denotes sets of states created from $\beta_i$, i.e., $\alpha_i = \{(a,\beta_i) \mid a \in A\}$.
\end{definition}

The Cartesian product over subset will be used to first alter the strategy graph.
The multiplied states shall defend the same set of vertices as before.
Then, during expansion, the guards shall be moved in order to defend new parts of the graph.
There, we need to ensure that the strategy remains defending.
For this, we have the following lemma that tackles the unchanged and changed states in separate cases.

\begin{lemma}\label{lem:copy_keeps_domination}
  Let us have $H = \GCPOS{G_1}{A}{G_2}$ with vertices labelled as in \Cref{def:copy}.
  \begin{enumerate}
    \item \label{lem:copy_keeps_domination:unchanged}
      If $C$ is a dominating set of $G_1$, then $\{(c,b) \mid (c,b) \in V(H), c \in C, b = \emptyset \lor b \in A\}$ is a dominating set of $H$.
    \item \label{lem:copy_keeps_domination:changed}
      If $A$ is a dominating set of $G_1$ and $B$ is a dominating set of $G_2$, then $A \times B$ (i.e., $\{(a,b) \mid a \in A, b \in B\}$) is a dominating set of $H$.
  \end{enumerate}
\end{lemma}
\begin{proof}
  Let us have a vertex $(x,y)$ in $V(H)$.

  In Case \ref{lem:copy_keeps_domination:unchanged}, there is a $c \in C$ that dominates $x$ in $G_1$.
  Therefore, $(c,y')$ with $y'=\emptyset$ if $c\not\in A$ or with $y'=y$ if $c\in A$ dominates $(x,y)$.

  In Case \ref{lem:copy_keeps_domination:changed}, if $x \not\in A$ (so $y=\emptyset$), then there is some $a \in A$ that dominates $x$ in $G_1$.
  As $|B| \ge 1$ there is $(a,b) \in A\times B$ that dominates $(x,y)$ in $H$.
  Otherwise $x \in A$ so there is $b \in B$ that dominates $y$ in $G_2$.
  Hence, $(x,y)$ is dominated by $(x,b) \in A\times B$ in $H$.
\end{proof}

When changing the strategy, we want to keep the properties of \Cref{lem:copy_keeps_domination} to ensure that \LAB is defending.

The following lemma shows the second major operation for changing strategies.
It allows us to add leaves to arbitrary vertex and defend the new graph with one more guard.

\begin{lemma}[Leaves addition]\label{lem:extend_leaves}
  Let us have a graph $G$ and let $u \in V(G)$ such that it has $\ell \ge 1$ adjacent leaf vertices $v_1,\dots,v_\ell$.
  Let $G'$ be a graph $G$ with vertices $v_1,\dots,v_\ell$ removed.
  For any defending \LAB $\mathcal{B}'$ there is a defending \LAB $\mathcal{B}$ with strategy graph $S_G = \GCPOS{S'_{G'}}{\vertexStates(u)}{K_\ell}$ that uses one more guard than $\mathcal{B}'$.
\end{lemma}
\begin{proof}
  First, let $S_{G'} = \GCPOS{S'_{G'}}{\vertexStates(u)}{\{\alpha_1,\dots,\alpha_\ell\}}$ (see short notation \Cref{def:short_copy_notation}).
  As $u$ must be defended $\vertexStates(u) \ne \emptyset$.
  By its construction, all guards of strategy $S_{G'}$ are stationary on $\TRAN(\alpha_i,\alpha_j)$.
  Let $\delta' = \Omega' \setminus \{\alpha_1,\dots,\alpha_\ell\}$.
  We expand the strategy over $S_{G'}$ to $G$ by adding $u$ to $\delta$ (i.e., $P(\delta) = P'(\delta') \cup \{u\}$) and adding $v_i$ to $\alpha_i$.
  We set $\TRAN(\alpha_i,\alpha_j) = \{(v_i,u),(u,v_j)\}$ and we extend $\TRAN(\delta,\alpha_i)$ with $(u,v_i)$.

  As $\alpha_i$ dominates the clique and $\vertexStates'(u)$ dominates $S'_{G'}$ we have by \Cref{lem:copy_keeps_domination} that $\mathcal{B}$ is a defending \LAB for $G$.
\end{proof}

\begin{observation}\label{obs:more_leaves}
  In \Cref{lem:extend_leaves} the construction works the same for any number of leaves, hence, we may add additional leaves after its use retroactively.
\end{observation}

We shall build strategies where vast majority of leaves are defended with \Cref{lem:extend_leaves}.
This gives a merit to treat all such states in the same way as their transitions with respect to the rest of the graph are isomorphic.
To do this we put forward the following notion.

\begin{definition}[Group state]\label{def:group_state}
  Let a \emph{group defense} be a set of states which were created by Cartesian product of $G$ and a clique $K_n$ over a subset.
\end{definition}

We shall use group defense only to describe groups of leaves.
By \Cref{obs:more_leaves} we will be able to add new leaves to such group at any point of the construction.

In group states, but also in general strategies we investigated, it seems that vertices which are adjacent to multiple leaves are often permanently occupied.
To get a concrete result from this observation let us show how to alter an \EDNinstance strategy such that such vertices are permanently occupied.

\begin{definition}[Permanently defended]\label{def:permanent}
  A vertex $u$ is \emph{permanently defended} (permanently occupied) in $\mathcal{B}$ if $\vertexStates(u) = \Omega$, i.e., $u \in P(\alpha)$ for every $\alpha \in \Omega$.
\end{definition}

\begin{lemma}\label{lem:always_present_guard}
  For a graph $G$ and its arbitrary defending \LAB $\mathcal{B}'$ we may create a defending \LAB $\mathcal{B}$ which uses the same number of guards and where each vertex adjacent to at least $2$ leaves is permanently defended.
\end{lemma}
\begin{proof}
  Let vertex $u$ be a vertex with at least $2$ adjacent leaves such that $u$ is not permanently occupied in in strategy $\mathcal{B}'$.
  Let $\alpha'$ be a state where no guard occupies $u$.
  In $P'(\alpha')$, there must be a guard on each leaf adjacent to $u$.
  Let $v$ and $w$ be two of the leaves adjacent to $u$ in $G'$.
  Let $S_{G'} = \GCPOS{S'_{G'}}{\Omega' \setminus \vertexStates'(u)}{\{\alpha_v,\alpha_w\}}$ (see short notation \Cref{def:short_copy_notation}).
  Let $P'(\alpha_v) = P(\alpha') \cup \{u\} \setminus \{w\}$ and $P'(\alpha_w) = P(\alpha') \cup \{u\} \setminus \{v\}$, so $P(\alpha_v)$ occupies $v$ and $P(\alpha_w)$ occupies $w$.

  To create transitions $\TRAN$, we keep all transitions between states within $\vertexStates'(u)$ the same.
  For all $\beta \in \Omega \setminus \vertexStates'(u)$, we set $\TRAN(\alpha_v, \beta)$ as $\TRAN(\alpha', \beta)$ with $w$ in each movement substituted by $u$.
  Such substitution still constitutes movements as $N[w] \subseteq N[u]$.
  Similarly, we create the transitions for $\TRAN(\alpha_w, \beta)$.
  Furthermore, we set $\TRAN(\alpha_w, \alpha_v) = \{(v,u),(u,w)\}$.
  Note, that the transitions in reverse direction are derived from symmetry.
  This shows that there are valid transitions for all edges of the created strategy graph $S_G$.

  In the obtained strategy $\mathcal{B}$ vertex $u$ is permanently defended.
  As we use Cartesian product with a complete graph over a dominating subset it follows from \Cref{lem:copy_keeps_domination} that the strategy is still defending.

  We repeat the above procedure for each vertex adjacent to at least $2$ leaves until all such vertices are permanently defended.
\end{proof}

\section{Reducing Cactus Graphs}\label{sec:reducing_cactus_graph}

In this section, we prove that $\EGC(G) = \EDN(G)$ for cactus graphs by showing optimal strategies and unconditional lower bounds.
The main idea is to repeatedly use reductions on the cactus graph $G$ to produce smaller cactus graph $G'$.
Then we prove that a strategy for $G$ uses a fixed number of guards more than an optimal strategy for $G'$.
Respective lower bound then shows that the strategy for $G$ is indeed also optimal.
We will describe precise way we get such results in \Cref{sec:technique_and_overview} but before that, we show the overall structure of the proof.

The proof uses an induction on the number of vertices.
The base case is a small graph ($1$ or $2$ vertices) where the optimal strategy is elementary (see \Cref{def:base_case}).
The induction step is described in detail later.
Now we show several structural properties of cactus graphs which allow us to do the induction.

\begin{definition}[Leaf cycle]\label{def:leaf_cycle}
  \emph{Leaf cycle} is a cycle which has at most one vertex (called \emph{connecting vertex}) which has neighbor such that it is not a vertex of the cycle nor a leaf.
\end{definition}
See a leaf cycle on \Cref{fig:white_pink_red}.

\begin{definition}[Leaf component]\label{def:leaf_component}
  By a \emph{leaf component} we mean either a leaf cycle or a leaf vertex which is not adjacent to a leaf cycle.
\end{definition}

\begin{observation}\label{obs:blockcut_decomposition}
  Every cactus graph with at least $3$ vertices contains a leaf component.
\end{observation}
\begin{proof}
  Let us obtain the block-cut tree $\mathcal{T}$ representation of the cactus graph and root it in an arbitrary block node (see \cite[block-cutpoint trees, page 36]{graph-theory}).
  Each block node of $\mathcal{T}$ either represents a single edge or a cycle.
  We observe the deepest nodes of $\mathcal{T}$ to get the following three cases, see \Cref{fig:block_cut_leaves}.
  \begin{itemize}
    \item[A] There is a deepest node which represents a cycle.
    \item[B] A deepest node's grandparent block is a single edge block.
    \item[C] A deepest node's grandparent block is a cycle block.
    \item[D] No deepest node has a grandparent.
  \end{itemize}
  \begin{figure}[h]
    \centering
    \includegraphics[scale=1.2]{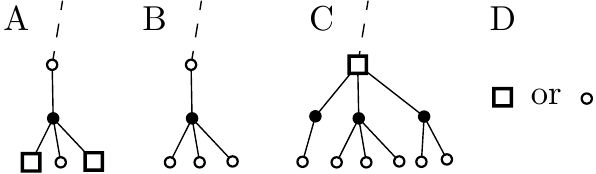}
    \caption[Block-cut tree substructures]{
      Example subtrees for structures which always appear in the block-cut tree of a cactus graph.
      Big squares represent cycle nodes, small full circles represent articulations, and small empty circles represent single edge nodes.}%
    \label{fig:block_cut_leaves}
  \end{figure}
  Note that in cases A and C the graph contains a leaf cycle, in case B it contains a leaf vertex which is not adjacent to a leaf cycle.
  The case D is trivial and the graph is either a cycle or a single edge.
  Hence, a cactus graph always contains a leaf component.
\end{proof}

\begin{definition}[Vertex colors]\label{def:vertex_colors}
  A non-connecting vertex $v$ of a leaf cycle $C$ is labeled with a color $\col(v)$ which depends on the number of adjacent leaves in the following way.
  \[
    \text{$\col(v)$ is\,}
    \begin{cases}
      \white      & \text{if $v$ is adjacent to $0$ leaves} \\
      \pink       & \text{if $v$ is adjacent to $1$ leaf} \\
      \red        & \text{if $v$ is adjacent to at least $2$ leaves}
    \end{cases}
  \]
  We shall label $v$ as $\col(v)=\connecting$ if $v$ is the connecting vertex of $C$.
  When a vertex can have different colors (to cover several cases at once) we list them by set of colors.
  For that purpose, we may write $\white$ instead of $\{\white\}$, and similarly for $\pink$, $\red$, and $\connecting$.
  Also, we use a shortcut to denote all colors $\unknown = \{\white,\pink,\red,\connecting\}$.
\end{definition}
See \Cref{fig:white_pink_red} for an example of $\white$, $\pink$, and $\red$ vertices.

\begin{figure}[h]
  \begin{minipage}[c]{0.47\textwidth}
    \centering
    \includegraphics[scale=1.2]{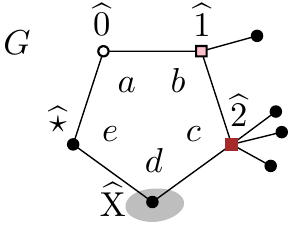}
  \end{minipage}
  \begin{minipage}[c]{0.52\textwidth}
    \caption[Example leaf cycle and vertex colors]{%
      A leaf cycle of a graph $G$ with a partial \LAB $\mathcal{B}$ containing vertices $a$, $b$, $c$, $d$, and $e$ with colors $\white$, $\pink$, $\red$, $\connecting$, and $\unknown$, respectively.
      The original graph was bigger and its rest was connected to $d$.
      To look at the leaf cycle in isolation, the graph was cut in vertex $d$ that now constitutes interface of $\mathcal{B}$.
      This cycle would be denoted by a leaf sequence $(\connecting,\unknown,\white,\pink,\red,\connecting)$ (or reversed).
    }%
    \label{fig:white_pink_red}
  \end{minipage}
\end{figure}

To describe reductions over leaf components we will use a concise notation for the leaf cycles which just lists the colors of consecutive vertices of the cycle as follows.
See \Cref{fig:white_pink_red} for an example.

\begin{definition}[Leaf sequence]\label{def:concise_reductions}
  Let $(v_1,\dots,v_n)$ be $n$ consecutive vertices of a leaf cycle.
  The \emph{leaf sequence} of vertices $(v_1,\dots,v_n)$ is $(\col(v_1),\dots,\col(v_n))$ where $\col(v_i) \subseteq \{\white,\pink,\red,\connecting\}$.
  Moreover, given two leaf sequences $A$ and $B$ and a graph $G$ which contains a leaf cycle with a leaf sequence $A$, let $A \to B$ denote a reduction of subgraph with leaf sequence $A$ to one with leaf sequence $B$ in $G$ to obtain $G'$.
\end{definition}

Note that if the leaf sequence starts and ends with a connecting vertex and contains no $\unknown$, then it describes the whole cycle because colors correspond to the number of leaves and there is only one connecting vertex in a leaf cycle.

Now, we show the base case and the overview of the induction step.

\begin{definition}[Base cases]\label{def:base_case}
  Let the \emph{base cases} be the following graphs along with their optimal defending \LABS.
  \begin{itemize}
    \item A single isolated vertex with no edges defended by \LAB
      \[
        \big( (\{u\},\emptyset),(\{\alpha\},\emptyset),\{\alpha \to \{u\}\},\emptyset\big).
      \]
    \item A single isolated edge is defended by \LAB
      \[
        \Big(
          \big(\{u,v\},\{\{u,v\}\}\big),
          \big(\{\alpha,\beta\},\{\{\alpha,\beta\}\}\big),
          \big\{\alpha \to \{u\},\beta \to \{v\}\big\},
          \big\{(\alpha,\beta) \to \{(u,v)\}\big\}
        \Big).
      \]
  \end{itemize}
\end{definition}

\subsection{Technique and Overview}\label{sec:technique_and_overview}

For the induction step, every reduction takes the cactus graph $G$ and changes it to $G'$ which has smaller number of vertices.
Reductions will be performed on a leaf component which by \Cref{obs:blockcut_decomposition} is always present in a cactus graph on at least three vertices.
The two cactus graphs which have at most two vertices are covered by base cases from \Cref{def:base_case}.

More precisely, every reduction shows lower bound and upper bound.
Lower bound is shown for the \EGCproblem and involves using \Cref{obs:identification,obs:lb_leaf,obs:lb_star,obs:lb_path}.
These tools make the graph smaller and show that in any defending strategy the removed parts required some minimum number of guards; they give us lower bound $\EGC(G) \geq \EGC(G') + K$ for some constant $K$.

Upper bounds are shown for the \EDNproblem and usually involve two separate steps.
First step takes the reduced graph $G'$ and its optimal strategy $S'_{G'}$ and shows how to alter the strategy by tools shown in \Cref{sec:upper_bound_tools}.
This does not change the number of guards, but only structure of the defense.
Second step uses the framework shown in \Cref{sec:upper_bounds}.
It takes part of the graph we intend to expand (\Cref{def:expansion}), cuts it, and replaces with an interface equivalent (\Cref{def:interface_equivalent}) partial \LAB, as described in \Cref{obs:upper_bound}.
During the expansion, the strategy graph $S_G$ does not change (so $S_G = S_{G'}$), however the graph and mapping does change.
The \LAB now maps strategy graph states so that there are new guards and some states have guards moved to other vertices of the graph.
We also show how the transitions change between states that were altered.
When the defense of the graph is managed with $K$ additional guards, this gives us an upper bound $\EDN(G) \leq \EDN(G')+K$ (the same $K$ as in the lower bound).

Combining the lower and upper bound using \Cref{lem:technique} results in an optimal number of guards for $G$ for \EDNproblem and \EGCproblem.

The used reduction depends on a leaf component that the cactus graph contains by \Cref{obs:blockcut_decomposition}.
If the deepest node is not adjacent to a leaf cycle, then we use leaf reductions shown in \Cref{sec:reducing_trees}.
Using these reduction exhaustively results in having a leaf cycle (or a base case) -- we will show this soon in \Cref{lem:tree_reductions_result}.

To reduce leaf cycles we will need additional properties on edges of the leaf cycle.
This involves being able to forbid movement along an edge, and forcing move along an edge.
We achieve this by partitioning all states of the strategy graph into tree groups which ensure these properties.
The properties are established in \Cref{sec:properties}.

Having the properties we take the leaf cycle and look at its vertex colors.
If a color pattern is listed among reductions then we have a way to remove it.
The reductions are split into two groups.
First recognizes just a small part of the cycle, making it shorter -- these are called cycle reductions \Cref{sec:cycle_reductions}.
The second recognizes the whole cycle and removes it entirely, leaving just a few leaves in its place -- these are constant component reductions shown in \Cref{sec:constant_reductions}.
We show that one of these reductions may always be used by exhaustive search of all possibilities in depicted in \Cref{fig:case_study}; and doubling this function, we show a slightly different proof in \Cref{lem:exhaustive_application_cycle}.

We end this section with the aforementioned proof of the cactus graph structure after application of leaf reductions in \Cref{lem:tree_reductions_result} and a diagram overview of the remaining sections in \Cref{fig:overview}.

\begin{lemma}\label{lem:tree_reductions_result}
  Exhaustive application of Reductions~\ref{reduction-t1}, \ref{reduction-t2}, and \ref{reduction-t3} on a cactus graph $G$ results in reaching the base case or it results in a cactus graph with a leaf cycle.
\end{lemma}
\begin{proof}
  We saw in \Cref{obs:blockcut_decomposition} that in every cactus there either is a leaf cycle or there is a set of $\ell \ge 1$ leaves with a common parent which is connected to the rest of the graph with a single edge.
  The number $\ell$ directly implies which tree reduction may be applied.
  If $\ell = 1$, then we may use Reduction~\ref{reduction-t1}; if $\ell = 2$, then we use Reduction~\ref{reduction-t3}; and last, if $\ell > 2$, then we use Reduction~\ref{reduction-t2}.
  After exhaustive application we either reach the base case or the other case applies -- we have a leaf cycle.
\end{proof}

\begin{figure}[h]
  \centering
  \input{images/overview.tex}
  \caption[Overview of \Cref{sec:reducing_cactus_graph}]{%
    Overview of \Cref{sec:reducing_cactus_graph}.
    Left boxes represent tools obtained in \Cref{sec:toolbox} (see \Cref{fig:toolbox}) and properties we introduce in \Cref{sec:properties};
    Right box shows structure of \Cref{sec:reducing_cactus_graph};
    Left-to-right arrows show which tools are used for which results.
    Right-to-right (green) arrows show that the reduction is partially based on or uses another reduction.
  }%
  \label{fig:overview}
\end{figure}
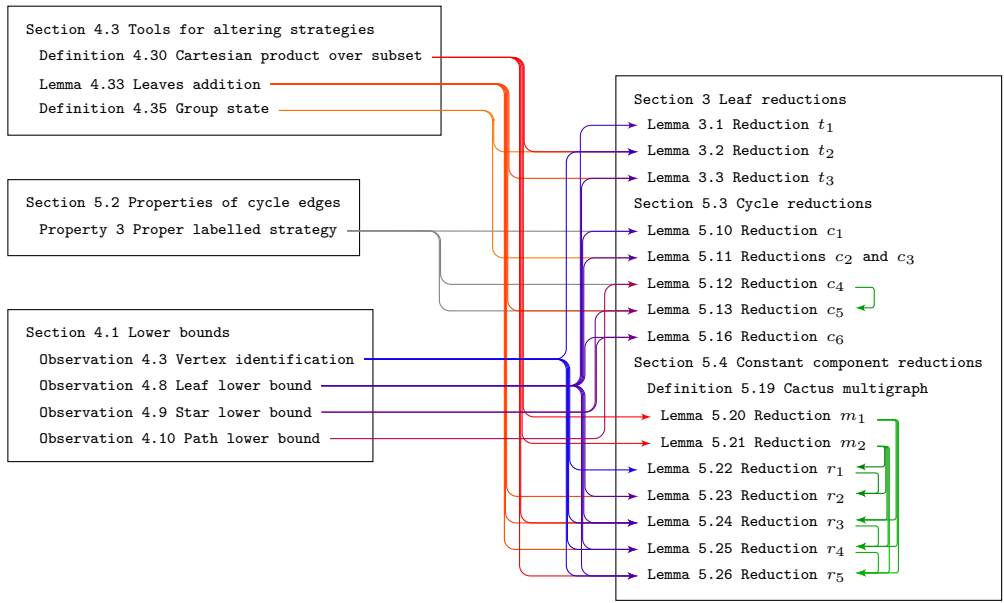

\subsection{Properties of Cycle Edges}\label{sec:properties}

We shall assume that the built strategy over the graph holds some properties which allow us to make stronger induction step.
More precisely, these properties shall be necessary to show Reductions~\ref{reduction-c1}, \ref{reduction-c4}, and \ref{reduction-c5}.

\begin{definition}[Edge states]\label{def:edge_states}
  By \emph{edge states of $(u,v)$ in $\Omega$} (where $\{u,v\} \in E(G)$) we mean creating sets $L_{u,v}$, $R_{u,v}$, and $N_{u,v}$ such that
  \begin{align*}
    \alpha \in L_{u,v}& \text{ if } \exists \beta\in\Omega,\; (u,v) \in \TRAN(\alpha,\beta), \\
    \alpha \in R_{u,v}& \text{ if } \exists \beta\in\Omega,\; (v,u) \in \TRAN(\alpha,\beta), \\
    N_{u,v}& = \Omega \setminus (L_{u,v} \cup R_{u,v})
  \end{align*}
\end{definition}

Note that because the orientation of the edge plays a role in these definitions, we have $L_{a,b}=R_{b,a}$, $R_{a,b}=L_{b,a}$, and $N_{a,b}=N_{b,a}$.
The names of the sets reflect from which side a guard can traverse the edge.
Also note, that if we assume symmetry then when moving to and from $N_{u,v}$ the edge $\{u,v\}$ cannot be traversed.
We propose the following edge property which is somewhat similar to \Cref{obs:defending_dominating_set}.

\begin{property}[Proper edge states]\label{prop:nonempty_edge_states}
  For a strategy $\mathcal{B}$ over graph $G$ an edge $\{u,v\}$ holds \emph{\Cref{prop:nonempty_edge_states}} if and only if its edge states $L_{u,v}$, $R_{u,v}$, and $N_{u,v}$ are all non-empty, $L_{u,v} \cap R_{u,v} = \emptyset$, and each of them is a dominating set over $S_G$.
\end{property}

There are several ramifications of an edge having \Cref{prop:nonempty_edge_states}.
Because of $L_{u,v} \cap R_{u,v} = \emptyset$ there is no state where we may choose to move over $(u,v)$ or $(v,u)$, i.e., at most one of these movements is available.
At the same time, as each of these sets is dominating $S_G$, it follows that we may get into any of these sets in one transition.
Last, as each set is non-empty we may force the strategy to forbid to move over $\{u,v\}$ in the current and one future transition by moving to $N_{u,v}$ at any point.
Additionally, we may force a movement over $(u,v)$ by moving first to some $\alpha \in L_{u,v}$ and then to $\beta \in R_{u,v}$ such that $(u,v) \in \TRAN(\alpha,\beta)$ as per \Cref{def:edge_states}.

All the properties that proper edge states additionally have compared to non-proper edge states are true irrespective of permutations of $L_{u,v}$, $R_{u,v}$, and $N_{u,v}$.
Hence, we may use the same sets on different edges by permuting them and checking that they constitute edge states of the new edge.

\begin{observation}\label{obs:edge_states_remapping}
  For edges $\{u,v\}$ and $\{a,b\}$ if we map proper edge states $L_{u,v}$, $R_{u,v}$, and $N_{u,v}$ to new sets $L_{a,b}$, $R_{a,b}$, and $N_{a,b}$ (with possibly permuting them) then these constitute proper edge states of $\{a,b\}$ if and only if they constitute edge states of $\{a,b\}$.
\end{observation}

\begin{proof}
  If the new edge states $L_{a,b}$, $R_{a,b}$, and $N_{a,b}$ do not constitute edge states then they trivially cannot be proper edge states.
  When $L_{u,v}$, $R_{u,v}$, and $N_{u,v}$ are proper edge states then they are disjoint and nonempty.
  These properties do not depend on their order so as long as the new states are edge states they will be proper.
\end{proof}

Our goal will be to have \Cref{prop:nonempty_edge_states} on all edges that lie on a leaf cycle that are incident to at least one $\white$ or $\connecting$ vertex.
We shall also show that it holds in some special cases to make several reductions easier.

In reductions, we will check that an edge has \Cref{prop:nonempty_edge_states}, however, the intuition about it is as follows.
We need to check whether each cycle edge is traversed at least once and whether it is not traversed at all by at least one state.
Also, it is usually trivial, but we should check that the edge cannot be traversed in both directions from some state.

\begin{property}[Proper labelled strategy]\label{prop:labelled_strategy}
  A partial \LAB $\mathcal{B}$ over a cactus graph $G$ has \Cref{prop:labelled_strategy} if and only if \Cref{prop:nonempty_edge_states} holds for each edge that lie on a cycle and
  \begin{itemize}
    \item is incident to a $\white$ or a $\connecting$ vertex,
    \item or is on a leaf cycle $(\connecting,\red,\red,\connecting)$,
    \item or is incident to a $\connecting$ vertex while not being a edge which lies between $\connecting$ and a $\red$ vertex on leaf cycle $(\connecting,\red,\white,\white,\red,\connecting)$ or $(\connecting,\white,\white,\red,\connecting)$.
  \end{itemize}
\end{property}

Our goal is to keep our cactus graph proper (as per \Cref{prop:labelled_strategy}) in all steps of reducing.
For simplicity, we shall work with reductions as if all edges on cycles which are incident to $\white$ or $\connecting$ vertex have \Cref{prop:nonempty_edge_states} and we shall tackle the exceptions to this rule separately in \Cref{obs:crrc_expansion} and \Cref{lem:rwwr_non_property}.

\subsection{Cycle Reductions}\label{sec:cycle_reductions}

Due to \Cref{lem:tree_reductions_result} we know that applying tree reductions may result in either solving the instance entirely or we obtain a leaf cycle.
In this section, we will tackle leaf cycles with cycle reductions which results in a leaf cycle of constant size.
Constant-sized leaf cycles are then resolved in \Cref{sec:constant_reductions}.

Let $C$ denote a leaf cycle where vertices are labeled with colors according to \Cref{def:vertex_colors}.
Cycle reductions consist of the following reductions (see notation in \Cref{def:concise_reductions}).
E.g., Reduction~\ref{reduction-c1} describes that a graph $G$ with a leaf cycle that contains consecutive vertices $U$ with colors $(\unknown,\pink,\unknown)$ may be changed to $G'$ by substituting $U$ with a vertices of colors $(\unknown,\unknown)$ (so just $\pink$ was removed).
At the same time, it claims that $\EDN(G) \leq \EDN(G')+1$ and $\EGC(G) \geq \EGC(G')+1$.
All of this is concisely written as $(\unknown,\pink,\unknown) \to (\unknown,\unknown) + 1$.

\begin{reduction}\customlabel{reduction-c1}{$c_1$}
  $(\unknown,\pink,\unknown) \to (\unknown,\unknown) + 1$ where $(\unknown,\unknown)$ has \Cref{prop:nonempty_edge_states}.
\end{reduction}
\begin{reduction}\customlabel{reduction-c2}{$c_2$}
  $(\red,\pink,\unknown) \to (\red,\unknown) + 1$
\end{reduction}
\begin{reduction}\customlabel{reduction-c3}{$c_3$}
  $(\red,\red,\unknown) \to (\red,\unknown) + 1$
\end{reduction}
\begin{reduction}\customlabel{reduction-c4}{$c_4$}
  $(\unknown,\white,\white,\white,\unknown) \to (\unknown,\unknown) + 1$ where $(\unknown,\unknown)$ has \Cref{prop:nonempty_edge_states}.
\end{reduction}
\begin{reduction}\customlabel{reduction-c5}{$c_5$}
  $(\unknown,\white,\red,\white,\unknown) \to (\unknown,\unknown) + 2$ where $(\unknown,\unknown)$ has \Cref{prop:nonempty_edge_states}.
\end{reduction}
\begin{restatable}{reduction}{restateredcsix}\customlabel{reduction-c6}{$c_6$}
  $(\connecting,\red,[\white,\red]^{2k},\connecting) \to (\pink) + 3k + 1$
  and
  $(\connecting,\red,[\white,\red]^{2k+1},\connecting) \to (\red) + 3k + 2$
\end{restatable}

Let $a$ and $b$ be the first and the last vertex of the leaf cycle in $G'$, respectively, that are described by the reduction.
It is clear that these reductions may be used in cases where $a$ and $b$ are non-connected disjoint vertices.
We note that the reductions will be used when $\{a,b\} \in E(G')$ though the result contains a pair of multiedges between $a$ and $b$ in $G$.
Moreover, these reductions may be used even in case where $a=b$.
Applying the reduction in such a case results in a loop in $a$ within $G'$.
Though loops and multiedges may be created by the process they will be immediately removed.
These cases will be addressed in \Cref{sec:loops_multiedges}.

Reductions~\ref{reduction-c1}, \ref{reduction-c4}, and \ref{reduction-c5} require the edge that is being expended (edge $\{a,b\}$ in $G'$) holds \Cref{prop:nonempty_edge_states}.
We shall ensure this by keeping \Cref{prop:labelled_strategy} for $G'$ while ensuring that during every expansion this property is preserved.

\begin{table}[tb]
  \caption[List of cycle reductions]{List of the cycle reductions; notation is the same as in \Cref{tab:leaf_reductions}}\label{tab:cycle_reductions}
  \centering
  \begin{tabular}{|c|c|c|}
    \hline
    Reduction & Lower bound & Upper bound \\
    \hline
    \reductioncase{c1}{4}
    \hline
    \reductioncase{c2}{5}
    \hline
    \reductioncase{c3}{6}
    \hline
    \reductioncase{c4}{7}
    \hline
    \reductioncase{c5}{8}
    \hline
    \ref{reduction-c6} & \multicolumn{2}{c|}{See \Cref{fig:reduction-c6-lb,fig:rwc_state_states}} \\
    \hline
  \end{tabular}
\end{table}

\begin{lemma}\label{lem:reduction-c1}
  \standardReductionLemma{c1}{1}
\end{lemma}
\begin{proof}
  Let us label the vertices of colors $(\unknown,\pink,\unknown)$ by $a,u,b$, respectively.
  Let $v$ be the leaf adjacent to $u$.
  By using \Cref{obs:lb_leaf} on vertices $\{u,v\}$ we get lower bound $\EGC(G) \geq \EGC(G')+1$.

  For upper bound, let $L_{a,b}$, $R_{a,b}$, and $N_{a,b}$ be edge states of the edge ${a,b}$ in the strategy of $G'$ obtained as stated in \Cref{def:edge_states}.
  We extend all states of $L_{a,b}$ and $R_{a,b}$ by adding $u$ to them, and we add $v$ to $N_{a,b}$.
  We substitute movements $(a,b)$ with $\{(a,u),(u,b)\}$ in $\TRAN(L_{a,b},R_{a,b})$ and we add $(u,v)$ to $\TRAN(L_{a,b} \cup R_{a,b}, N_{a,b})$.
  The new vertices are defended as $L_{a,b}$ and $N_{a,b}$ are dominating the strategy graph because \Cref{prop:nonempty_edge_states} holds for $\{a,b\}$ in $G'$.
  The edge states for the new edges $\{a,u\}$ and $\{u,b\}$ in $G$ remain the same as for $\{a,b\}$ in $G'$.
  Therefore, these edges now hold \Cref{prop:nonempty_edge_states} in $G$.
  By extending all states with one guard we got a defending \LAB, so $\EDN(G) \leq \EDN(G')+1$.
  By \Cref{lem:technique} we get that $G$ is defended with one more guard than $G'$.
\end{proof}

Reductions~\ref{reduction-c2} and \ref{reduction-c3} merge a group of consecutive red and pink vertices and defend leaves adjacent to them by a group state (\Cref{def:group_state}).

\begin{lemma}\label{lem:reduction-c23}
  Let $G'$ be $G$ after application of Reduction~\ref{reduction-c2} or~\ref{reduction-c3}.
  $G$ is defended with $1$ more guard than $G'$.
\end{lemma}
\begin{proof}
  The reductions are separate for the sake of future argument but they are proven in the same way.
  Let us label the vertices of colors $(\red,\{\pink,\red\},\unknown)$ by $a,u,b$, respectively.
  Let $R_1$ denote all leaves adjacent to $u$.
  By applying \Cref{obs:lb_leaf} on $R_1 \cup \{u\}$ we get lower bound $\EGC(G) \geq \EGC(G')+1$.

  For upper bound, let $\gamma'$ be the group state for leaves adjacent to $a$.
  We add to leaves $R_1$ two new vertices which are defended by $\gamma'$ as pointed out in \Cref{lem:extend_leaves}.
  Now we split $u$ from $a$, taking its leaves with it that we now label by $R_2$.
  Transitions between leaves is extended to $\TRAN(R_1,R_2) = \{(R_1,a),(a,u),(u,R_2)\}$ and similarly, we extend all transitions which used $a$.
  The transitions that interacted with $a$ and $b$ are preserved, so the reduction expands interface equivalent partial \LABS.
  Though we did not need \Cref{prop:nonempty_edge_states} the graph still has \Cref{prop:labelled_strategy} because the new edge $\{u,b\}$ takes on exact transitions that $\{a,b\}$ had.
  So if $\{a,b\}$ held the property in $G'$, then $\{b,u\}$ holds it in $G$.

  We added one guard so $\EDN(G) \leq \EDN(G')+1$ and by \Cref{lem:technique} we get that $G$ is defended with one more guard than $G'$.
\end{proof}

\begin{lemma}\label{lem:reduction-c4}
  \standardReductionLemma{c4}{1}
\end{lemma}
\begin{proof}
  Let us label the vertices of colors $(\unknown,\white,\white,\white,\unknown)$ by $a,c,u,d,b$, respectively.
  Using \Cref{obs:lb_path} on $\{c,u,d\}$ we get lower bound $\EGC(G) \geq \EGC(G')+1$.

  For upper bound, let $L'_{a,b}$, $R'_{a,b}$, and $N'_{a,b}$ be edge states of the edge ${a,b}$ in the strategy of $G'$ obtained from \Cref{def:edge_states}.
  These are proper edge states as $\{a,b\}$ holds \Cref{prop:nonempty_edge_states} in $G'$.
  We extend the states by adding $d$ to all states of $L'_{a,b}$, $c$ to $R'_{a,b}$, and $u$ to $N'_{a,b}$; this creates sets $L_{a,b}$, $R_{a,b}$, and $N_{a,b}$.
  We substitute movements along $(a,b)$ with $\{(a,c),(d,b)\}$ in $\TRAN(L_{a,b},R_{a,b})$, hence, the exchanged parts of the graph are interface equivalent.
  We add $(c,u)$ to $\TRAN(R_{a,b}, N_{a,b})$ and $(d,u)$ to $\TRAN(L_{a,b}, N_{a,b})$.
  The new $\{c,u,d\}$ vertices are defended by the nonempty sets $R_{a,b}$, $N_{a,b}$, and $L_{a,b}$, respectively.
  The edge states for the new edges are as follows.
  \begin{align}\label{eq:c4_new_states}
    L_{a,b} &= L_{a,c} = L_{d,b} = N_{c,u} = L_{d,u} \nonumber \\
    R_{a,b} &= R_{a,c} = R_{d,b} = L_{c,u} = N_{d,u} \\
    N_{a,b} &= N_{a,c} = N_{d,b} = R_{c,u} = R_{d,u} \nonumber
  \end{align}
  As these are only permutations of the edge sets by \Cref{obs:edge_states_remapping} they hold \Cref{prop:nonempty_edge_states}.
  We get $\EDN(G) \leq \EDN(G')+1$.
  By \Cref{lem:technique} we get that $G$ is defended with one more guard than $G'$.
\end{proof}

\begin{lemma}\label{lem:reduction-c5}
  \standardReductionLemma{c5}{2}
\end{lemma}
\begin{proof}
  The proof goes very similarly as the proof of \Cref{lem:reduction-c4}, but all states $N_{a,b}$ shall group defend leaves adjacent to $u$ while $u$ will be permanently occupied.

  We label the vertices of colors $(\unknown,\white,\red,\white,\unknown)$ by $a,c,u,d,b$, respectively.
  Let $R$ be the leaves neighboring $u$.
  Using \Cref{obs:lb_star} on $u$ and its neighborhood we get lower bound $\EGC(G) \geq \EGC(G')+2$.

  Repeat the same sequence of steps as in the proof of \Cref{lem:reduction-c4} which uses one guard and then add leaves adjacent to $u$ by \Cref{lem:extend_leaves} using one extra guard.
  This does not change transitions over the edges which are not incident to the leaves so by \Cref{obs:edge_states_remapping} they still hold \Cref{prop:nonempty_edge_states}.
  We get $\EDN(G) \leq \EDN(G')+2$.
  By \Cref{lem:technique} we get that $G$ is defended with two more guards than $G'$.
\end{proof}

We remark that using reductions \ref{reduction-t1}, \ref{reduction-t2}, \ref{reduction-t3}, \ref{reduction-c1}, \ref{reduction-c4}, and a small set of constant component reductions is sufficient to solve so-called Christmas cactus graphs (graphs where each edge is in at most one cycle and each vertex is in at most two $2$-connected components) for which the optimal strategy we presented in~\cite{Blazej2019EDN}.
The remaining reductions tackle vertices of color $\red$, which are not present in the class of Christmas cactus graphs.

The last cycle reduction is a curious special case, let us recall it first.

\restateredcsix*

We shall use all the other cycle reductions first and if none of them can be used, then we use Reduction~\ref{reduction-c6}.
This allows us to assume a particular structure which we define and prove now.
The reason behind this structure may also be well understood from decision diagram of reduction application in \Cref{fig:case_study}.

\begin{definition}[RW-cycle]\label{def:rw_cycle}
  A leaf cycle $C$ is a \emph{RW-cycle} if it consists of vertices with alternating $\red$ and $\white$ colors such that the first and last is $\red$, i.e., $C = (\connecting,\red,\white,\red,\white,\dots,\red,\white,\red,\connecting)$.
\end{definition}

\begin{lemma}\label{lem:exhaustive_application}
  Assume a leaf cycle $C$ where Reductions~\ref{reduction-c1}, \ref{reduction-c2}, \ref{reduction-c3}, \ref{reduction-c4}, and \ref{reduction-c5} cannot be applied anywhere.
  Then, $|C| \leq 6$ or $C$ is a \emph{RW-cycle}.
\end{lemma}
\begin{proof}
  First, note that if the leaf cycle contains $\pink$ vertices, then there always is a $\pink$ vertex which neighbors $\white$ or $\connecting$, in that case we use Reduction~\ref{reduction-c1}, or it neighbors $\red$, we use Reduction~\ref{reduction-c2}.
  Hence, all $\pink$ vertices are removed if Reductions~\ref{reduction-c1} and~\ref{reduction-c2} were exhaustively used.

  Next, as only $\red$, $\white$, and $\connecting$ vertices remain, exhaustively using Reduction~\ref{reduction-c3} ensures that there are no two adjacent $\red$ vertices.
  (Note that if only red vertices remained, than we end up with $(\connecting,\red,\connecting)$ which removes the multiedge by Reduction~\ref{reduction-m2} and then uses Reduction~\ref{reduction-t3}.)

  Last, if there is a $(\red,\white,\white)$ part of a leaf cycle then the $\red$ vertex is either also adjacent to $\connecting$ or Reduction~\ref{reduction-c5} can be used as would $(\unknown,\white,\red,\white,\white)$ necessarily occur.
  Exhaustively using Reduction~\ref{reduction-c4} ensures that such cases do not occur.
  So whenever there are two $\white$ adjacent vertices the $\connecting$ vertex is at distance at most $2$ from them.
  This means that either the cycle has at most $5$ vertices or the vertices of colors $\red$ and $\white$ alternate and constitute a RW-cycle.
\end{proof}

The lower bound for an RW-cycle will not be much harder than for other reductions, however, for the upper bound we will need a strategy made just right for such a cycle.

\begin{lemma}\label{lem:reduction-c6}
  Reduction~\ref{reduction-c6} is correct.
\end{lemma}
\begin{proof}
  Let us label the vertices along the cycle as $u_1, u_2, \dots, u_n$ with $u_1$ being the connecting vertex (so that $\red$ vertices are even).
  Let $u_{n+1}=u_1$.
  Let $R_2$, $R_4$, etc. be the leaves adjacent to the red vertices $u_2$, $u_4$, and so on.

  First, we use \Cref{obs:lb_leaf} on $\{u_6\} \cup R_6$; second, we use \Cref{obs:lb_star} on $N[u_4]$.
  We shortened the cycle by $4$ and got lower bound of $3$.
  By repeating the argument $k$ times we end up with a cycle $G'$ of constant size $2$ or $4$.
  The cycle of size $2$ gets reduced by \ref{reduction-m2} and then \ref{reduction-t3} which results in a lower bound of $1$.
  The RW-cycle of size $4$ has form $(\connecting,\red,\white,\red,\connecting)$ and its lower bound is shown in Reduction~\ref{reduction-r4} to be $2$.
  Putting the $3k$ for every $4$ vertices together with $1$ and $2$ lower bound for respective sizes of the cycle, we get the desired lower bounds.
  See \Cref{fig:reduction-c6-lb} for an illustration of shortening the cycle by $4$.
  \begin{figure}[h]
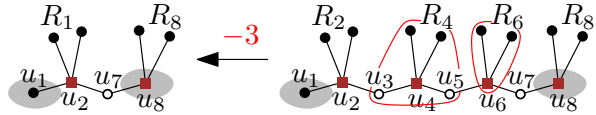

    \centering
    \reductionimage{lower}{9}
    \caption{Part of the lower bound proof for Reduction~\ref{reduction-c6}}%
    \label{fig:reduction-c6-lb}
  \end{figure}

  Strategy for the upper bound is quite tricky to describe so let us define a few new notions just for its description.
  Let \emph{red-parity} of an even number $i$ be the parity of $i/2$.
  This divides red vertices of the RW-cycle into red-odd and red-even, based on the red-parity.
  Let \emph{reverse labeling} be the labeling of the RW-cycle in opposite ordering, i.e., if $u'_1=u_{n+1}$, $u'_2=u_n$, $u'_3=u_{n-1}$, \dots, $u'_n=u_2$, and $u'_{n+1}=u_1$, then $u'_1,\dots,u'_n$ is the reverse labeling with respect to labeling $u_1,\dots,u_n$.

  For the upper bound, we distinguish two cases depending on the size of the RW-cycle.
  First case is that the size is $n = 4k + 2$, the second has size $n = 4k$.

  We now focus on the first case, where the RW-cycle has size $n = 4k + 2$.
  Note that in RW-cycle of this size reverse labeling does not change the red-parity of red vertices.
  We alter the strategy by gradually expanding the states as follows.
  \[
    S^*_{G'} = \GCPOS{S'_{G'}}{\mathcal{S}'(u_1)}\{\alpha_1, \alpha_2, \beta_4, \beta_8, \dots, \beta_{4k}\}
    \text{~~and~~}
    S_{G'} = \GCPOS{S^*_{G'}}{\Omega \setminus \mathcal{S}'(u_1)}\{\beta_2, \beta_6, \dots, \beta_{4k+2}\}
  \]
  Now we perform expansion from $G'$ to $G$ and set the states $\Omega$ of $S_G$ as follows.
  \begin{align}
    P(\alpha_1) &= \bigcup_{x=0}^{k} \{ u_{4x+3} \}
    \qquad
    P(\alpha_2) = \bigcup_{x=0}^{k} \{ u_{4x+1} \} \nonumber
    \\
    P(\beta_{4x}) &= \big((P(\alpha_2) \cap \{u_j\}_{j<i}) \cup (P(\alpha_1) \cap \{u_j\}_{j>i})\big) \cup \{R_{4x}\} \label{eq:reduction_c6_states_definition}
    \\
    P(\beta_{4x+2}) &= \big((P(\alpha_1) \cap \{u_j\}_{j<i}) \cup (P(\alpha_2) \cap \{u_j\}_{j>i})\big) \cup \{R_{4x+2}\} \setminus \{u_1\} \nonumber
  \end{align}
  Notice that $u_1 \in P(\alpha_2)$ as $u_{4x+1}=u_1$ for $x=0$ and $u_1 \in P(\alpha_1)$ as $u_{4x+3}=u_{4k+3}=u_{n+1}=u_1$ for $x=k$.
  For $P(\beta_{4x+2})$ the intersections imply that $u_1$ is not contained, but we mention it explicitly for clarity (as we do not consider $u_{n+1}$ to be $u_j$ for $j < i$ even though it equals $u_1$).
  The state $\beta_{2i}$ group defends leaves adjacent to $u_{2i}$.
  Note that they behave differently based on their their red-parity.

  Now for the transitions.
  To make the notation concise let us shorten consecutive movements through red vertices as $F_i = \{(u_i, u_{i+1}), (u_{i+1}, u_{i+2})\}$ and $B_i = \{(u_i, u_{i-1}), (u_{i-1}, u_{i-2})\}$ (as forward and backward).
  Note that we use $F_i$ and $B_i$ only for odd values of $i$.

  Let us have integers $x$ and $y$ and assume, without loss of generality, that $x \leq y$.
  We set the movements of transitions as follows.
  If $2x$ and $2y$ have the same red-parity, then
  \begin{equation}\label{eq:same-red-parity}
    \TRAN(\beta_{2x}, \beta_{2y}) = \{(R_{2x}, u_{2x}), (u_{2x}, u_{2x+1})\} \cup
    \bigcup_{i = x}^{y-4} F_{2i+3}
    \cup \{(u_{2y-1}, u_{2y}), (u_{2y}, R_{2y})\}.
  \end{equation}
  Otherwise, $2x$ and $2y$ have different red-parity.
  If $x$ is odd (so $y$ is even), then their transition is defined as follows.
  \[
    \TRAN(\beta_{2x}, \beta_{2y}) = \{(R_{2x}, u_{2x}), (u_{2x}, u_{2x-1})\} \cup
    \bigcup_{i = 2}^x B_{2i-3} \cup
    \!\bigcup_{i=y+1}^{2k-1}\!\!\! B_{2i+3} \cup
    \{(u_{2y+1}, u_{2y}), (u_{2y}, R_{2y})\}
  \]
  If the red-parity is different and $x$ is even, then in reverse labeling and swapping $x$ with $y$ we end up in the case where the red-parity is still different, but $x$ is odd.
  This case was already solved.
  The other direction of these transitions is filled in by symmetry (\Cref{prop:strategy_symmetry}).

  It remains to describe transitions with $\alpha_1$ and $\alpha_2$.
  \begin{align}
    \TRAN(\alpha_1, \alpha_2) &= \{(u_1, u_1)\} \cup \bigcup_{i=0}^{k-1} F_{4i+3} \nonumber
    \\
    \TRAN(\beta_{4x}, \alpha_1) &= \{(R_{4x}, u_{4x}), (u_{4x}, u_{4x-1})\} \cup \bigcup_{i=1}^{x-1} B_{4i+1} \label{eq:even-to-alpha-1}
    \\
    \TRAN(\beta_{4x+2}, \alpha_2) &= \{(R_{4x+2}, u_{4x+2}), (u_{4x+2}, u_{4x+1})\} \cup \bigcup_{i=1}^{x} B_{4i-1} \label{eq:odd-to-alpha-2}
  \end{align}
  Note that $\alpha_1$ is $\alpha_2$ in reverse labeling.
  Hence, the case $\TRAN(\beta_{4x}, \alpha_2)$ is equivalent to $\TRAN(\beta_{4x}, \alpha_1)$ in reverse labeling.
  Similarly, the case $\TRAN(\beta_{4x+2}, \alpha_1)$ is equivalent to $\TRAN(\beta_{4x+2}, \alpha_2)$ in reverse labeling.
  See a part of this strategy on \Cref{fig:rwc_state_states}.

  \begin{figure}[h]
    \centering
    \includegraphics[page=18]{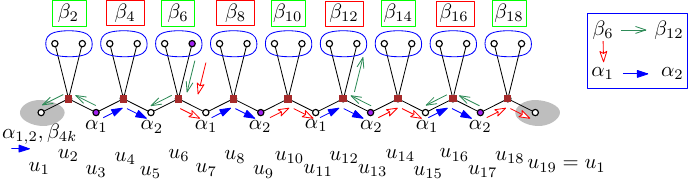}
    \caption[Example part of an RW-cycle strategy]{
      Part of a strategy on a RW-cycle of red-odd size $18$ with guards (shown purple) placed on $P(\beta_6)$.
      A few selected transitions are shown as an example.
    }%
    \label{fig:rwc_state_states}
  \end{figure}

  For the interface equivalency, note that in $\alpha_1,\alpha_2$, and $\beta_{4k}$ occupy $u_1$ and states $\beta_{4k+2}$ do not occupy $u_1$.
  We showed how to transition between every pair of states, so the strategies are interface equivalent with a single pink vertex with expanded states as in $S_{G'}$.

  \medskip
  Now for the second case, where the RW-cycle has size $n = 4k$.
  Note that in RW-cycle of this size reverse labeling changes the red-parity of red vertices, which was not true in the first case.
  The difference in the construction of the strategy is that now we expand from a red vertex.
  Let $u'_1$ and $u'_2$ be the two leaves of the red vertex.
  Let $\delta' = \Omega' \setminus (\mathcal{S}(u'_1) \cup \mathcal{S}(u'_2))$.
  We gradually alter the strategy in the following way.
  \begin{align*}
    S^1_{G'} &= \GCPOS{S'_{G'}}{\mathcal{S}(u'_1)}\{\gamma, \beta_4, \beta_8, \beta_{12}, \dots, \beta_n\}
    \\
    S^2_{G'} &= \GCPOS{S^1_{G'}}{\mathcal{S}(u'_2)}\{\alpha_1, \beta_2, \beta_6, \dots, \beta_{n-2}\}
    \\
    S_{G'} &= \GCPOS{S^2_{G'}}{\delta'}\{\alpha_2\}
  \end{align*}
  We perform the expansion to get $S_G$ such that all the states have exactly the same definitions as in the first case, see \Cref{eq:reduction_c6_states_definition}.
  We note a major difference: in the second case, $u_1$ is not an element of $P(\alpha_1)$.
  We added one extra state $\gamma$ which has $P(\gamma) = P(\alpha_1)$.
  There will be a major significance for this state when proving edge properties.

  Now we describe the transitions for the the strategy on $G$.
  For $2x$ and $2y$ of the same red-parity, the transition \Cref{eq:same-red-parity} still holds.
  In case $2x$ and $2y$ (with $x < y$) have different red-parity, then we consider two separate cases based on red-parity of $2x$.
  \begin{align*}
    \TRAN(\beta_{4x}, \beta_{4y+2}) = &\;\{(R_{4x}, u_{4x}), (u_{4x}, u_{4x-1})\}
    \cup \bigcup_{i=1}^{x-1} B_{4i+1} \cup \bigcup_{i=y+2}^k B_{4i-1}\;\cup
    \\
    \cup &\;\{(u_{4y+1}, y_{4y+2}, (u_{4y+2}, R_{4y+2})\}
    \\
    \TRAN(\beta_{4x+2}, \beta_{4y}) = &\;\{(R_{4x+2}, u_{4x+2}), (u_{4x+2}, u_{4x+1})\}
    \cup \bigcup_{i=0}^{x-1} B_{4i+3} \cup \bigcup_{i=y+2}^{k+1} B_{4i-3}\;\cup
    \\
    \cup &\;\{(u_{4y-1}, y_{4y}, (u_{4y}, R_{4y})\}
  \end{align*}
  Notice the difference in $u_1$ -- transition $\TRAN(\beta_{4x}, \beta_{4y+2})$ does not move through $u_1$ so there $u_1$ is stationary during it; in $\TRAN(\beta_{4x+2}, \beta_{4y})$ movements $\{(u_2,u_1),(u_1,u_n)\}$ happen.
  To fill all possibilities of mutual transitions among $\beta_{2x}$ we add transitions obtained by reversed labeling and symmetry.

  Now we show the transitions with $\alpha_1$ and $\alpha_2$.
  Note that reversed labeling does not change these two states.
  For $\beta_{4x+2}$ we can apply \Cref{eq:odd-to-alpha-2} to get $\TRAN(\beta_{4x+2},\alpha_2)$, and by reversing the labeling this gives us also $\TRAN(\beta_{4x}, \alpha_2)$.
  Note that after this transition there is one less guard on $G$ as it leaves through the interface $\{u_1\}$.
  In particular, $\TRAN(\beta_{2x},\alpha_2)$ moves to $u_1$ via $(u_2,u_1)$ if $2x$ is red-odd, and via $(u_n,u_{n+1})$ if $2x$ is red-even.

  Similarly, for $\beta_{4x}$ we can apply \Cref{eq:even-to-alpha-1} to get $\TRAN(\beta_{4x},\alpha_1)$, and by reversing the labeling we get $\TRAN(\beta_{4x+2}, \alpha_1)$.
  This transition did not interact with the interface.
  The transition among the two states is as follows.
  \[
    \TRAN(\alpha_1, \alpha_2) = \bigcup_{i=0}^{k-1} F_{4i+3}
  \]
  Note that this again results in a move $(u_n,u_{n+1})$.

  Last, we introduce the new state $\gamma$ which has the same guard configuration as $\alpha_1$, but differs in one transition.
  So $\TRAN(\gamma,\beta_{2x})=\TRAN(\alpha_1,\beta_{2x})$, and $\TRAN(\gamma,\alpha_2)=\emptyset$ (all guards are stationary), but $\TRAN(\gamma,\alpha_2)$ shall be $\TRAN(\alpha_1, \alpha_2)$ in reverse labeling.
  More precisely,
  \[
    \TRAN(\gamma,\alpha_2)=\bigcup_{i=0}^{k-1} B_{4i+3}.
  \]
  This contains a move $(u_2,u_1)$.
  See how movements interact with the interface in \Cref{fig:rwc_move_through_interface}.

  \begin{figure}[h]
    \centering
    \includegraphics[page=21,width=0.7\textwidth]{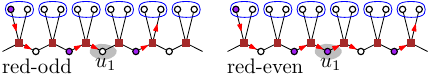}
    \caption[Movement through the interface in RW-cycles]{Movements through the $\connecting$ vertex in red-odd and red-even RW-cycle.}%
    \label{fig:rwc_move_through_interface}
  \end{figure}

  We discussed the interface impact of all transitions and note that they are equivalent to those in $S_{G'}$, hence, the exchanged strategy is interface equivalent.

  \medskip
  It remains to show that $S_G$ is a proper strategy in both cases.
  The strategy started $S'_{G'}$ was a clique and by Cartesian product over single vertices it remained a clique.
  Thus, it suffices to say that there is at least one state in $L_{u_i,u_{i+1}}$, $R_{u_i,u_{i+1}}$, and $N_{u_i,u_{i+1}}$, and that $L_{u_i,u_{i+1}} \cap R_{u_i,u_{i+1}} = \emptyset$ for every $i \in \{1, \dots, n\}$, as any non-empty subset of vertices of the clique is dominating.

  Now we show the partitioning of the states into $L_{u_{x}, u_{x+1}}, R_{u_{x}, u_{x+1}}$, and $N_{u_{x}, u_{x+1}}$ for each $x \in \{1, \dots, n\}$, see \Cref{fig:rwc-state-parts} for an illustration.
  First, observe that all closed neighborhoods of $u_{2x}$ for $4 \le 2x \le n-2$ contain exactly $2$ guards in all the states we defined for this strategy.
  Let $x$ be an even integer such that $4 \le x \le n-2$.
  Let $e = \{u_{x+1},u_{x+2}\}$.
  \begin{figure}[h]
    \centering
    \includegraphics[page=20,width=0.7\textwidth]{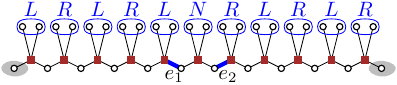}
    \caption[Part of the proper edge states proof]{The states of $\beta_{2x}$ that belong to $L_e, R_e, N_e$ for $e$ equal to the edges $e_1$ and $e_2$.}%
    \label{fig:rwc-state-parts}
  \end{figure}

  We show that $\beta_x \in N_e$ by a contradiction.
  Assume that in some transition from $\beta_x$ a guard moved through $e$.
  As in $\beta_x$ vertex $u_{x+1}$ is not occupied the guard must have moved from $u_{x+1}$.
  However, then $N[u_x]$ would have $3$ guards after the transition which cannot happen as we observed; a contradiction.

  For edges that could not be addressed in the argument because they are too close to the $\connecting$ vertex -- $e_1 = \{u_1,u_2\}$ and $e_2 = \{u_3,u_4\}$.
  We observe that for red-even RW-cycles $\alpha_1 \in N_{e_1}$ and $\gamma \in N_{e_2}$.
  For red-odd RW-cycles $\beta_2 \in N_{e_2}$ and $\beta_n \in N_{e_1}$.

  Now we claim that for any even $x$ such that $2 \le x \le n$, $e = \{u_{x+1},u_{x+2}\}$, the states $\beta_y$ where $y \ne x$ are in $L_e$ if and only if $u_{x+1} \in P(\beta_y)$, and they are in $R_e$ otherwise.
  We shall prove this more intuitively, as otherwise the claim can be proved by exhaustively listing all edges in all the transitions.
  First notice, that all movements from $\beta_y$ which are not incident to leaves are performed over a continuous part of the cycle which starts in $u_y$, and that they move ``away'' from $u_y$ towards the other end of the part.
  The movements always move an occupied $\white$ to $\red$ and if the part continues then moves the $\red$ to the adjacent $\white$ (this is true even when moving through the $\connecting$ vertex).
  Hence, if $e$ (that is not incident to $\connecting$) is included in the part of the movement, then we move $(u_{x+1},u_{x+2})$ if and only if $x_{x+1}$ is occupied, and we move $(u_{x+2},u_{x+1})$ if and only if $u_{x+1}$ was unoccupied.
  This proves the claim.

  Remainder of the edges which start at even positions and their $N$, $L$, and $R$ sets can be obtained by the same argument on reversed labelling.

\end{proof}

\subsection{Constant Component Reductions}\label{sec:constant_reductions}

The following lemma shows that considering the constant component cases completes the list of all necessary reductions.

\begin{lemma}\label{lem:exhaustive_application_cycle}
  Let us have a cactus graph $G$.
  After an exhaustive application of leaf and cycle reductions the reduced cactus graph $G'$ is either a base case or it contains a leaf cycle of constant size.
\end{lemma}
\begin{proof}
  First, by \Cref{obs:blockcut_decomposition} the cactus always contains a leaf component.
  If we exhaustively apply tree reductions, then by \Cref{lem:tree_reductions_result} we are either done or there is a leaf cycle $C$.
  In \ref{lem:exhaustive_application} we saw that an exhaustive application of the cycle rules results either in a base case or a cycle with alternating $\red$ and $\white$ vertices, which gets tackled by Reduction~\ref{reduction-c6}.
  The cases that remain are cycles of constant sizes where none of the reductions may be applied.
\end{proof}

We obtain the list of constant leaf cycles by the following procedure.
First, we apply Reductions~\ref{reduction-c1} and~\ref{reduction-c2} exhaustively.
This removes all pink vertices from the leaf cycle.
Now, let us scan over the vertices of the leaf cycle in a linear order of vertices along the cycle, starting from the connecting vertex.
On the one hand, whenever there is a cycle reduction applicable on the vertices which were scanned so far, then we can apply it.
Hence, such a leaf cycle does not belong to constant leaf cycle cases.
On the other hand, when the cycle returns back to the connecting vertex and still no cycle reduction may be used, then this cycle constitutes a constant leaf cycle.
We present a full search diagram in \Cref{fig:case_study}.

\begin{figure}[tb]
  \centering
  \includegraphics[scale=1.14]{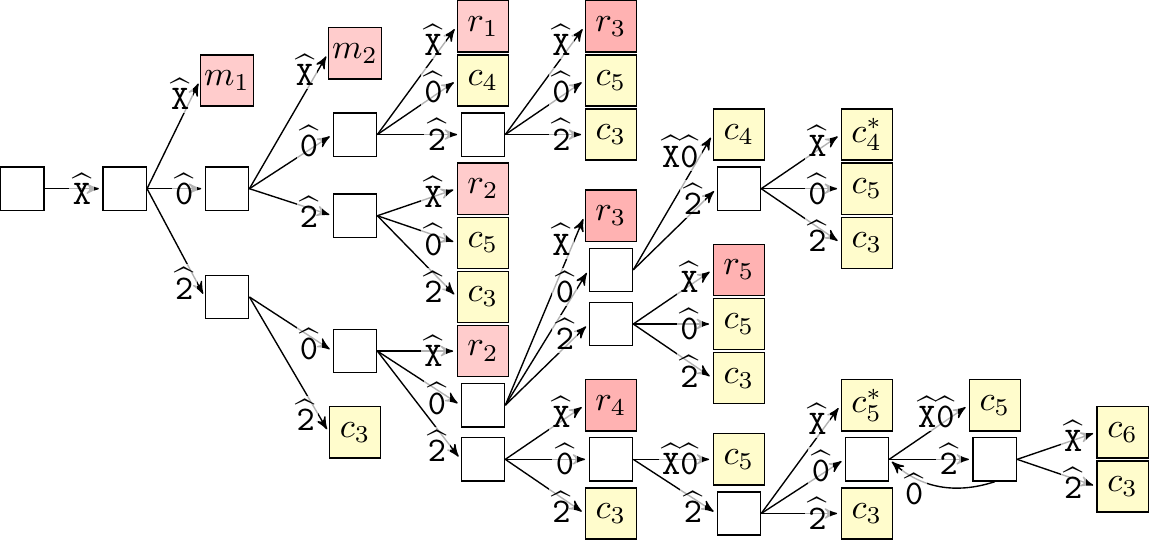}
  \caption[Case analysis of leaf cycle reductions]{
    Case analysis of applied reductions on a leaf cycle.
    Vertices $\pink$ were removed first by exhaustively applying their reductions.
    Scanning over vertices of a leaf cycle in order from the connecting vertex we identify these cases.
    The leaves show which reduction should be used for the scanned leaf cycle.
    Labels $c_1$ up to $c_5$ (yellow leaves) signify cycle reductions;
    labels $m_i$ and $r_i$ (red leaves) signify constant component reductions;
    nodes with a star ${}^{*}$ require \Cref{obs:crrc_expansion}.
    We can check that all the cases are covered by seeing that all inner (empty) nodes have outgoing edges labelled $\white$, $\red$, and $\connecting$.
  }%
  \label{fig:case_study}
\end{figure}

Again, we shall denote the reductions concisely as defined by \Cref{def:concise_reductions}.
However, in constant component reductions the leaf sequence describes the whole leaf cycle and the connecting vertex is listed as the first and the last vertex.
The vertices of the leaf cycle will be denoted by $u,u_1,u_2,\dots,u_{n-1},u$ where $u$ is the connecting vertex.
Let $R_1,\dots,R_{n-1}$ denote sets of all leaves adjacent to vertices $u_1,\dots,u_{n-1}$, respectively.
Note that the size $0 \le |R_i| \le 2$ and directly coincides with color of respective vertex $u_i$.
See \Cref{fig:crrc} for an example of a leaf sequence of constant leaf cycle and notation of its vertices.

Recall that the cycle reductions may be used even when the result does not create a simple graph, which is resolved in \Cref{sec:loops_multiedges}.

\begin{observation}\label{obs:crrc_expansion}
  A strategy for a leaf cycle $(u,u_1,u_2,u)$ with colors $(\connecting,\red,\red,\connecting)$ is built in such a way that the edge $(u_1,u_2)$ holds \Cref{prop:nonempty_edge_states} (even though it is not incident to a $\white$ vertex) which makes an expansion of Reductions~\ref{reduction-c4} or~\ref{reduction-c5} over this edge possible.
\end{observation}
\begin{proof}
  This leaf cycle gets reduced by Reduction~\ref{reduction-c3}, then \ref{reduction-m2}, and last with tree reduction \ref{reduction-t3}.
  We show that in the strategy resulting for expansions holds \Cref{prop:nonempty_edge_states} on edges $\{u,u_1\}$ and $\{u,u_2\}$.
  See \Cref{fig:crrc} for an illustration.
  Checking the exact movements of this strategy, we have that
  \begin{equation*}
    (u_1,u_2) \not\in \TRAN(L_{u,u_1},R_{u,u_1}),
    (u_1,u_2) \not\in \TRAN(L_{u,u_1},N_{u,u_1}),
    (u_1,u_2)     \in \TRAN(R_{u,u_1},N_{u,u_1}).
  \end{equation*}
  In particular, we may set $L_{u,u_1} = N_{u_1,u_2}$, $R_{u,u_1} = L_{u_1,u_2}$, and $N_{u,u_1} = R_{u_1,u_2}$.
  As the edge move sets for $\{u,u_1\}$ holds the properties which require all these sets to be non-empty, we have that they hold for $\{u_1,u_2\}$ as well.
  \begin{figure}[h]
    \centering
    \includegraphics[scale=1.2]{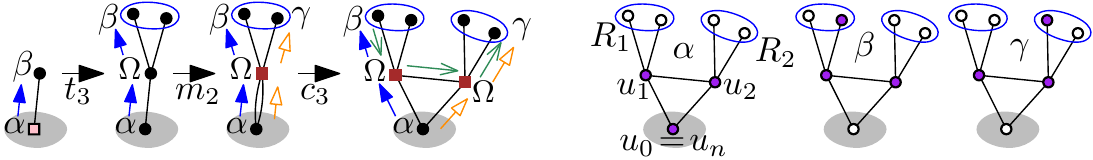}
    \caption[Strategy for $(\connecting,\red,\red,\connecting)$]{
      \textbf{Left:}
      Building the strategy for a $(\connecting,\red,\red,\connecting)$ leaf cycle.
      The states $\alpha$, $\beta$, and $\gamma$ are representants of sets $L_{u,a}$, $R_{u,a}$, and $N_{u,a}$, respectively.
      \textbf{Right:}
      Example guard configurations for states $\alpha$, $\beta$, and $\gamma$.
    }%
    \label{fig:crrc}
  \end{figure}
\end{proof}

From \Cref{obs:crrc_expansion} we know that the cases $(\connecting,\red,\white,\white,\white,\red,\connecting)$ and $(
\connecting,\red,\white,\red,\white,\red,\connecting)$ can be reduced by Reductions~\ref{reduction-c4} and~\ref{reduction-c5}, respectively.
See these cases in \Cref{fig:case_study}.

\subsubsection{Loops and Multiedges}\label{sec:loops_multiedges}

Similarly to \Cref{obs:crrc_expansion}, for the constant cases where we need to show that the properties hold.
By allowing cycle reductions to apply in cases where the vertices $a$ and $b$ are adjacent, or even identical, we allowed the result of the reduction to contain multiedges or loops.
This intermediate form of the graph can be thought of as a generalized cactus graph.

\begin{definition}[Cactus multigraph]\label{def:cactus_multigraph}
  Let the \emph{cactus multigraph} be a multigraph (possibly with loops) that is connected and its every edge lies on at most one cycle.
\end{definition}

A cactus multigraph differs from a cactus graph by allowing loops on arbitrary vertices (cycles of size $1$) and allowing $2$ multiedges between some vertices (cycles of size $2$).
The cactus multigraph may be changed to a cactus graph by removing multiedges and loops.
The following two reductions take care of that.

\begin{reduction}\customlabel{reduction-m1}{$m_1$}
  Let $G'$ be $G$ with one loop removed.
\end{reduction}
\begin{reduction}\customlabel{reduction-m2}{$m_2$}
  Let $G'$ be $G$ with a multiedge $\{u,v\}$ ($2$ edges) where $v$ has degree $2$ ($1$ neighbor) changed to a single edge.
\end{reduction}

\begin{figure}[h]
  \centering
  \includegraphics[page=1,scale=1.1]{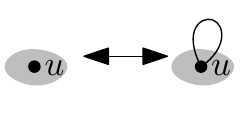}
  \qquad
  \includegraphics[page=2,scale=1.1]{reduction/loops.pdf}
  \caption[Loop and multiedge reduction]{
    \textbf{Left:}
    loop reduction \ref{reduction-m1};
    \textbf{Right:}
    multiedge reduction \ref{reduction-m2}
  }%
  \label{fig:loops}
\end{figure}

Observe these reductions on \Cref{fig:loops}.
We now prove that they do not need any additional guards.

\begin{lemma}\label{lem:reduction-m1}
  \standardReductionLemma{m1}{0}
\end{lemma}
\begin{proof}
  The strategy on $G$ can be easily adapted to $G'$ by replacing any guard movement along the loop of $u$ by not moving the guard on $u$, thus $\EGC(G') \leq \EGC(G)$.
  At the same time, any strategy on $G'$ is applicable on $G$, so $\EDN(G) \leq \EDN(G')$.
  The equality follows from \Cref{lem:technique}.
  However, we would like the loop in $u$ to have \Cref{prop:nonempty_edge_states}.

  Intuitively, to keep the properties, we could say that at any configuration where $u$ is occupied the guard can be moved along the loop in any direction or to be forbidden from moving along it while the configuration stays the same.
  Formally, we can achieve the same by setting $S_G = \GCPOS{S'_{G'}}{\vertexStates'{u}}{\{\alpha,\beta,\gamma\}}$.
  We now set that $\TRAN(\alpha,\beta)=\{(u,u)\}$.
  This creates $L_{u,u}=\alpha$, $R_{u,u}=\beta$, and $Q_{u,u}=\Omega\setminus\{\alpha,\beta\}$.
  This altered strategy holds \Cref{prop:nonempty_edge_states} for the loop of $u$ as the sets $L_{u,u}$, $R_{u,u}$, and $N_{u,u}$ are non-empty and dominating $S_G$ because $\vertexStates'(u)$ dominates $S'_{G'}$.
\end{proof}

In our case, Reduction~\ref{reduction-m1} gets used after Reduction~\ref{reduction-c4} is used on $(\connecting,\white,\white,\white,\connecting)$ or after Reduction~\ref{reduction-c5} is used on $(\connecting,\white,\red,\white,\connecting)$.
It could also be used on $(\connecting,\pink,\connecting)$ after Reduction~\ref{reduction-c1}; but in that case we can remove the multiedge first.

\begin{lemma}\label{lem:reduction-m2}
  \standardReductionLemma{m2}{0}
\end{lemma}
\begin{proof}
  Let $e_1, e_2$ be the two different edges $\{u,v\}$ oriented as $(u,v)$ in $G$.
  We assume that $G'$ is $G$ with $e_2$ removed.
  Lower and upper bound are clear as every move along $e_2$ can be changed to a move along $e_1$ and the strategy on $G'$ is applicable to $G$ without change.
  The challenge is, again, to show that \Cref{prop:nonempty_edge_states} holds for $e_1$ and $e_2$ in $G$.

  Let $\beta' = \vertexStates'(v)$ and $\alpha' = \Omega' \setminus \beta$.
  To prove the property on $e_1$ and $e_2$, we will modify the strategy on $G'$ in the following way.
  If $\beta' \ne \Omega'$, then there is a move along $e_1$ in $G'$.
  In that case, we set $S_{G'} = \GCPOS{S'_{G'}}{\beta'}{\{\beta,\gamma\}}$ while we alter the movements $\TRAN(\alpha,\gamma)$ to move along $e_2$ instead of $e_1$.
  The edge states have $\alpha \in L_{e_1}$, $\beta \in R_{e_1}$, and $\gamma \in N_{e_1}$, and similarly for $e_2$ (with swapped $\beta$ and $\gamma$).

  Second case is that $\beta' = \Omega'$ while $\alpha' \ne \Omega'$.
  Here, we alter the strategy such that for all states where $u$ is not occupied, we move the guard from $v$ to $u$.
  This makes it so that $v$ is occupied in states $\alpha$ which we now split into $\alpha_1$ and $\alpha_2$ in the same way as in the previous case.

  The last case is $\beta' = \Omega'$ while $\alpha' \ne \Omega'$.
  Here we set $S_{G'} = \GCPOS{S'_{G'}}{\Omega'}{\{\alpha_1,\alpha_2,\alpha_3\}}$ and setting $\TRAN(\{\alpha_1,\alpha_2\}) = \{(u,v),(v,u)\}$, i.e., transitioning along $e_1$ and $e_2$ in opposite directions.
  Also $\TRAN(\alpha_1,\alpha_3)$ and $\TRAN(\alpha_2,\alpha_3)$ have all guards stationary.
  This makes edge states as $\alpha_1 \in L_{e_1}$, $\alpha_2 \in R_{e_1}$, and $\alpha_3 \in N_{e_1}$ while the exact same edge states work for $e_2$.

  In all the cases the edge states are non-empty, hence, \Cref{prop:nonempty_edge_states} holds for $e_1$ and $e_2$ after Reduction~\ref{reduction-m2}.
  \begin{figure}[h]
      \centering
      \includegraphics[scale=1.2]{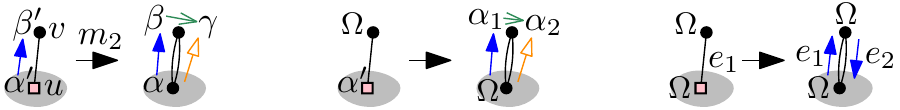}
      \caption[Cases of multiedge reduction]{
        Cases of Reduction~\ref{reduction-m2}.
        \textbf{Left:}
        There is a movement along the edge.
        \textbf{Middle:}
        Leaf is permanently occupied.
        \textbf{Right:}
        Leaf and its neighbor are permanently occupied.
      }%
      \label{fig:redm2}
  \end{figure}
\end{proof}
We note that in our strategy the case where $v$ is permanently defended shall not occur.

If we did not use Reduction~\ref{reduction-m2} the number of constant size leaf cycle reductions would be significantly bigger.
It gets used after reduction of $(\connecting,\white,\pink,\connecting)$ or $(\connecting,\pink,\pink,\connecting)$ by \ref{reduction-c1}, $(\connecting,\red,\pink,\connecting)$ by \ref{reduction-c2}, $(\connecting,\red,\red,\connecting)$ by \ref{reduction-c3}, $(\connecting,\white,\white,\white,\white,\connecting)$ or $(\connecting,\white,\white,\white,\red,\connecting)$ by \ref{reduction-c4}, $(\connecting,\white,\white,\red,\white,\connecting)$ or $(\connecting,\white,\red,\white,\red,\connecting)$ by \ref{reduction-c5}.
Without Reduction~\ref{reduction-m2} each of these cases would have to be analyzed separately.

\subsubsection{Constant Size Leaf Cycle Reductions}\label{sec:constant}

By \Cref{lem:exhaustive_application_cycle} the last cases that have to be resolved are covered by the following reductions.
See \Cref{tab:constant_reduction} for accompanying lower bound and upper bound proof illustrations.
Also see \Cref{fig:overview} for diagram of notions used within proofs of these reductions.

\begin{reduction}\customlabel{reduction-r1}{$r_1$}
  $(\connecting,\white,\white,\connecting) \to (\pink) + 0$
\end{reduction}
\begin{reduction}\customlabel{reduction-r2}{$r_2$}
  $(\connecting,\white,\red,\connecting) \to (\pink) + 1$
\end{reduction}
\begin{reduction}\customlabel{reduction-r3}{$r_3$}
  $(\connecting,\white,\white,\red,\connecting) \to (\red) + 1$
\end{reduction}
\begin{reduction}\customlabel{reduction-r4}{$r_4$}
  $(\connecting,\red,\white,\red,\connecting) \to (\red) + 2$
\end{reduction}
\begin{reduction}\customlabel{reduction-r5}{$r_5$}
  $(\connecting,\red,\white,\white,\red,\connecting) \to (\red) + 2$
\end{reduction}

\begin{table}[h]
  \caption[List of constant component reductions]{List of constant component reductions; thick red edges do not hold \Cref{prop:nonempty_edge_states}.}\label{tab:constant_reduction}
  \centering
  \begin{tabular}{|c|c|c|}
    \hline
    Reduction & Lower bound & Upper bound \\
    \hline
    \reductioncase{r1}{10}
    \hline
    \reductioncase{r2}{11}
    \hline
    \reductioncase{r3}{12}
    \hline
    \reductioncase{r4}{13}
    \hline
    \reductioncase{r5}{14}
    \hline
  \end{tabular}
  \break
\end{table}

Now we proceed to show correctness of these reductions.
First group consists of reductions where a leaf cycle is reduced to $\pink$ vertex $u$ and its leaf $v$.
The vertices of the expended leaf cycle are denoted by $u,u_1,u_2,\dots,u_{n-1},u$.

\begin{lemma}\label{lem:reduction-r1}
  \standardReductionLemma{r1}{0}
\end{lemma}
\begin{proof}
  Using \Cref{obs:identification} to identify $u_2$ with $u_1$ then using Reduction~\ref{reduction-m2} results in lower bound of $0$.

  For the upper bound, we first expand $\{u,v\}$ to multiedges $e_1$ and $e_2$ as per Reduction~\ref{reduction-m2}.
  Then we take $G'$ and change it to $G$ by splitting $v$ into two vertices $u_1$ and $u_2$.
  We create $\beta$ by substituting all occurrences of $v$ in $P(\beta')$ with $u_1$, and create $\gamma$ by substituting all occurrences of $v$ in $P(\gamma')$ with $u_2$.
  The transition between them becomes $\TRAN(\beta,\gamma)=\{(u_1,u_2)\}$.
  The strategy is interface equivalent as the strategy did not change states or transitions of the interface.

  We set $L_{u_1,u_2}=N_{u,u_1}$, $R_{u_1,u_2}=N_{u,u_1}$, and $N_{u_1,u_2}=L_{u,u_1}$ so the new edge $\{u_1,u_2\}$ holds \Cref{prop:nonempty_edge_states} and the strategy for $G$ holds \Cref{prop:labelled_strategy}.

  No guard was added so $\EDN(G) \leq \EDN(G')$ and by \Cref{lem:technique} we get that $G$ is defended with the same number of guards as $G'$.
\end{proof}

We recall that by $R_i$ we denote all leaves adjacent to $u_i$.

\begin{lemma}\label{lem:reduction-r2}
  \standardReductionLemma{r2}{1}
\end{lemma}
\begin{proof}
  Using \Cref{obs:lb_leaf} on $\{u_2\} \cup R_2$ then using Reduction~\ref{reduction-m2} results in lower bound of $1$.

  We do the same expansion as in the proof of \Cref{lem:reduction-r1}.
  After that, we use \Cref{lem:extend_leaves} to add leaves $R_2$ to $u_2$ while using one extra guard to defend it.
  Graph holds \Cref{prop:labelled_strategy} by the same argument as in the proof of \Cref{lem:reduction-r1}.
  We added one extra guard which results in $\EDN(G) \leq \EDN(G')+1$ and by \Cref{lem:technique} we get that $G$ is defended with one more guard than $G'$.
\end{proof}

We now prove correctness of the other three cases.
The reduced graph $G'$ now consists of a single $\red$ vertex $u$ (and its leaves).
The partial \LAB on $G'$ has states $\alpha'$ and $\beta'$ that defend the two leaves adjacent to $u$.
Also, let $\delta' = \Omega' \setminus (\alpha' \cup \beta')$, which may be an empty set.

\begin{lemma}\label{lem:reduction-r3}
  \standardReductionLemma{r3}{1}
\end{lemma}
\begin{proof}
  Using \Cref{obs:lb_leaf} on $u_3$ and one of its leaves, identifying $u_2$ with $u$ using \Cref{obs:identification}, and using Reductions~\ref{reduction-m1} and~\ref{reduction-m2} to remove loops and multiedges results in lower bound of $1$.

  For the upper bound, let $u_1$ and $u_3$ be the two leaves adjacent to $u$ in $G'$.
  Let $\alpha' = \vertexStates'(u_1)$ and $\beta' = \vertexStates'(u_3)$.
  We make $S_{G'} = \GCPOS{S'_{G'}}{\beta'}{\{\beta,\gamma\}}$.
  Now we expand the graph $G'$ by first applying \Cref{lem:extend_leaves} on $u_3$, adding $2$ new leaves to it using one additional guard.
  Next, we add a vertex $u_2$ while connecting it to $u_1$ and $u_3$ and we move $\gamma$ from $R_3$ to $u_2$ which is easy as $u_2$ is a neighbor of $u_3$.
  The only major change in transitions is that $\TRAN(\alpha,\gamma)=\{(u_1,u_2),(u,u),(u_3,u_3)\}$ instead of $\{(u_1,u),(u,u_3),(u_3,u_2)\}$.
  No other transitions change, and $u$ behaves the same, so the exchanged graphs are interface equivalent.
  See \Cref{fig:redconnecting_cwwrc} for strategy $S_G$.

  \begin{figure}[h]
      \centering
      \includegraphics[scale=1.2,page=1]{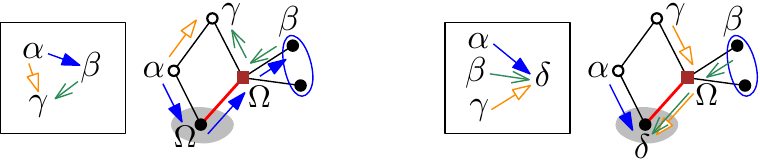}
      \caption{Strategy for $(\connecting,\white,\white,\red,\connecting)$ leaf cycle}%
      \label{fig:redconnecting_cwwrc}
  \end{figure}

  We note that each is traversed at some point and that $\alpha \in N_{u_2,u_3}$, $\beta \in N_{u_1,u_2}$, and $\gamma \in N_{u,u_1}$ so these edges hold \Cref{prop:nonempty_edge_states} and the strategy for $G$ holds \Cref{prop:labelled_strategy}.
  The edge $\{u,u_3\}$ does not need to hold the property as it is a special case tackled in \Cref{lem:rwwr_non_property}.

  We got that $\EDN(G) \leq \EDN(G')+1$ and by \Cref{lem:technique} we get that $G$ is defended with one more guard than $G'$.
\end{proof}

\begin{lemma}\label{lem:reduction-r4}
  \standardReductionLemma{r4}{2}
\end{lemma}
\begin{proof}
  Using \Cref{obs:lb_leaf} first on $\{u_1, v_1\}$ where $v_1 \in R_1$, then again on $\{u_3, v_3\}$ where $v_3 \in R_3$, identifying $u_2$ with $u$ using \Cref{obs:identification}, and using Reductions~\ref{reduction-m1} and~\ref{reduction-m2} to remove loops and multiedges results in lower bound of $2$.

  For upper bound, repeat exactly the expansion from \Cref{lem:reduction-r3} on $G'$ which uses one extra guard.
  Continue by applying \Cref{lem:extend_leaves} on $u_1$ which adds the leaves $R_1$ using one extra guard while returning the defending \LAB on $G$.
  The properties for edges $\{u_1,u_2\}$, $\{u_2,u_3\}$, and interface equivalency still hold from \Cref{lem:reduction-r3}.
  However, we can split $\gamma$ into two states $\gamma_1$ and $\gamma_2$ which dictates whether $\TRAN(\gamma_i,\delta)$ traverses through $\{(u_2,u_1),(u_1,u)\}$ or $\{(u_2,u_3),(u_3,u)\}$.
  This ensures \Cref{prop:labelled_strategy} for $\{u,u_1\}$ and $\{u,u_3\}$ as former cannot be traversed from $\gamma_2$ and latter from $\gamma_1$.
  Hence, we have $\EDN(G) \leq \EDN(G')+2$ and by \Cref{lem:technique} we get that $G$ is defended with two more guards than $G'$.
\end{proof}

\begin{lemma}\label{lem:reduction-r5}
  \standardReductionLemma{r5}{2}
\end{lemma}
\begin{proof}
  Using \Cref{obs:lb_leaf} first on $\{u_1, v_1\}$ where $v_1 \in R_1$, then again on $\{u_4, v_4\}$ where $v_4 \in R_4$, and last identifying $u_2$ and $u_3$ with $u$ using \Cref{obs:identification} results in lower bound of $2$.

  For upper bound, repeat exactly the expansion from \Cref{lem:reduction-r4} on $G'$ which uses two extra guards (we do not use part of the proof which proved the property).
  Then we make $S_{G'} = \GCPOS{S'_{G'}}{\vertexStates'(u_2)}{\{\gamma_1,\gamma_2\}}$, i.e., splitting $\gamma$ into $\gamma_1$ and $\gamma_2$.
  We expand $G'$ to $G$ by splitting $u_2$ into $u_2$ and $u_3$ (while renaming $u_3$ to $u_4$).
  We preserve a guard of $\gamma_1$ on $u_2$ and $\gamma_2$ on $u_3$.
  Transition between them will be $\TRAN(\gamma_1,\gamma_2) = \{(u_2,u_3),(u,u)\}$.
  This is interface equivalent.
  See \Cref{fig:redconnecting_crwwrc} for strategy $S_G$.

  \begin{figure}[h]
      \centering
      \includegraphics[scale=1.15,page=2]{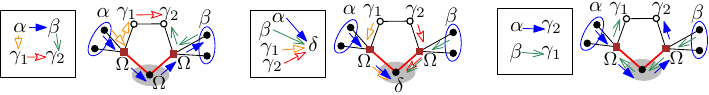}
      \caption{Strategy for $(\connecting,\red,\white,\white,\red,\connecting)$ leaf cycle}%
      \label{fig:redconnecting_crwwrc}
  \end{figure}

  We have \Cref{prop:labelled_strategy} as each edge is traversed and $\{u_1,u_2\}$ cannot be traversed from $\gamma_2$, $\{u_2,u_3\}$ from $\alpha$, and $\{u_3,u_4\}$ cannot be traversed from $\gamma_1$.
  We note that the other two cycle edges $\{u,u_1\}$ and $\{u,u_4\}$ are part of the exception which is tackled in \Cref{lem:rwwr_non_property}.
  Hence, we have $\EDN(G) \leq \EDN(G')+2$ and by \Cref{lem:technique} we get that $G$ is defended with two more guards than $G'$.
\end{proof}

Now we tackle the exception in \Cref{prop:labelled_strategy} which influences Reductions~\ref{reduction-r3} and~\ref{reduction-r5}.

\begin{lemma}\label{lem:rwwr_non_property}
  The order of reductions can be changed so that in a $(\connecting,\white,\white,\red,\connecting)$ or $(\connecting,\red,\white,\white,\red,\connecting)$ leaf cycle \Cref{prop:nonempty_edge_states} is not required for edges that connect a $\connecting$ and a $\red$ vertex.
\end{lemma}
\begin{proof}
  Let us label by $e$ an edge which connects a $\connecting$ and a $\red$ vertex.
  Reductions which require the \Cref{prop:nonempty_edge_states} on an edge are Reductions~\ref{reduction-c1},~\ref{reduction-c4}, and~\ref{reduction-c5}.
  If $e$ is not a result of any of these reductions then there is no need for $e$ to hold \Cref{prop:nonempty_edge_states}.
  Otherwise, let us analyze the cases separately.
  \begin{itemize}
    \item
      Reduction~\ref{reduction-c1} resulted in $e$ -- before reduction we had $(\connecting,\pink,\red,\dots)$ where we can use Reduction~\ref{reduction-c2} instead.
      This results in $(\connecting,\red,\dots)$ without needing the property for $e$.
    \item
      Reduction~\ref{reduction-c4} resulted in $e$ -- before reduction we had $(\connecting,\white,\white,\white,\red,\white,\white,[\red,]\connecting)$.
      Hence, we may use Reduction~\ref{reduction-c5} instead.
      This results in $(\connecting,\white,\white,\white,[\red,]\connecting)$ where $e$ has the property.
    \item
      Reduction~\ref{reduction-c5} resulted in $e$ -- before reduction we had $(\connecting,\white,\red,\white,\red,\white,\white,[\red,]\connecting)$ so we may use Reduction~\ref{reduction-c5} on the second $\red$ vertex instead.
      This results in $(\connecting,\white,\red,\white,[\red,]\connecting)$ where $e$ has the property.
  \end{itemize}
  We used other reductions to avoid reaching these leaf components by reductions that would require \Cref{prop:nonempty_edge_states}.
  The first described case can be used at any point.
  The last two described cases are used on constant leaf components and as the result is different, it follows that their edges hold the property.
\end{proof}

This concludes the constant component reductions which together with cycle components and approach described in \Cref{sec:technique_and_overview} give us a polynomial algorithm to solve \EDNproblem.

\section{Future Work}\label{sec:future_work}
The presented tools could be useful in a future study of the \EDNproblem on different graph classes. 
For instance, grids of size $\{3, 5\} \times n$ were extensively studied~\cite{Messinger2017,5n-grids}.
We believe it would interesting to see to which extent the tools could by applied in study of grids of less restricted dimensions.

Another noteworthy class of graphs are the so called dually chordal graphs, for which many domination related problems are polynomial time solvable.
It would be interesting to see whether \EDNproblem remains polynomial time solvable as well.

Furthermore, the computational complexity of the decision variant of the \MED\ problem is still mostly unknown as mentioned in the introduction.
It remains open whether the problem is in \PSPACE and whether it is \PSPACE-hard.




\bibliographystyle{plain}
\bibliography{main}


\appendix

\section{Complete Strategies}

We note that if the strategy $S_G$ was a complete graph, then strategy $S'_G$ created by the application of \Cref{lem:always_present_guard} is also a complete graph.

\begin{property}\label{prop:complete_strategy}
  A partial \LAB $\mathcal{B} = (G, S_G, P, \TRAN, R)$ is \emph{complete} if $S_G$ is a complete graph, i.e., there is $\{\alpha,\beta\}\in\FF$ for every $\alpha,\beta \in \Omega$.
\end{property}

We note that there are graphs where every optimal strategy is not complete, see \Cref{apx:sec:noncomplete} for such an example.
Complete strategies can be effectively pruned to contain at most $|V(G)|$ states in the following way.

\begin{lemma}
  For any complete defending \LAB of cardinality $k$ with the minimum number of vertices of $S_G$ it holds $|V(S_G)| \le |V(G)| - k + 1$.
\end{lemma}
\begin{proof}
  Pick an arbitrary complete defending strategy $S_G$ which uses $k$ guards.
  For each $v \in V(G)$ we shall pick one state $\alpha_v \in V(S_G)$ such that $v \in P(\alpha)$.
  First, pick any $\alpha \in V(S_G)$ and assign it as state to each $v \in P(\alpha)$.
  Then, for every $v \in V(G) \setminus P(\alpha)$ assign $\alpha_v \in \vertexStates(v)$ as its state.
  We just picked $|V(G)| - k + 1$ states such that they form a strategy where every pair of states is traversable and which is defending as it covers all the vertices of $G$.
\end{proof}

Similar to completeness of a strategy we may talk about the graph class of $S_G$ to describe its properties.

\section{Non-complete Strategy}\label{apx:sec:noncomplete}

\begin{observation}\label{obs:no_clique}
  An optimal \EDNinstance strategy on $5 \times 5$ grid uses at least $7$ guards.
\end{observation}
\begin{proof}
  Let us denote vertices of the grid by $u_{i,j}$ where $1 \le i,j \le 5$.

  First, we show a lower bound of $7$.
  Assume for a contradiction that there is a defending strategy $S_6$ with at most $6$ guards.
  Any state of $S_6$ needs to dominate all $25$ vertices.
  There must exist a state $C$ where $u_{2,2}$ is occupied.
  In $C$ there also must be at least one guard in the closed neighborhood of each corner ($u_{1,1}$, $u_{5,1}$, $u_{1,5}$, and $u_{5,5}$).
  In the grid a vertex may dominate at most $5$ vertices and a vertex on the side of the grid may dominate at most $4$ vertices.
  All vertices in the closed neighborhood of corners are on the side of the grid.
  Additionally, vertex which dominates $u_{1,1}$ may dominate at most $2$ new vertices, as $u_{2,2}$ already dominates many of vertices in its neighborhood.
  In total, the $6$ guards of $C$ may dominate at most $2\cdot 5 + 3 \cdot 4 + 2 = 24$, a contradiction.
\end{proof}

The upper bound can be shown by construction of a strategy, however, we have no good tools to show that all the strategies are not complete graphs -- we found this using a full strategy-space search.
A construction which uses $7$ guards contains three states and majority of their reflections and rotations, see them on \Cref{fig:grid_strategy}.
In this case, we do not show the strategy, as it contains roughly $20$ states (depending on a slight optimization, it may be less) that would contain $190$ transitions.

\begin{figure}[h]
  \centering
  \includegraphics[page=1,width=0.20\textwidth]{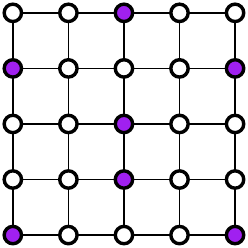}
  \hspace{1cm}
  \includegraphics[page=2,width=0.20\textwidth]{5x5-grid/grid-strategy.pdf}
  \hspace{1cm}
  \includegraphics[page=3,width=0.20\textwidth]{5x5-grid/grid-strategy.pdf}
  \caption[Strategy on a $5 \times 5$ grid]{%
    The $5 \times 5$ grid has $19$ m-eternal dominating sets.
    Each of the configurations can be expressed as a combination of rotations and reflections of exactly one of these $3$ basic configurations.
    Each of the $19$ configurations is necessary for the strategy and can move into at most $12$ other states.
  }%
  \label{fig:grid_strategy}
\end{figure}





\end{document}

%% file: src/fancyProblem.tex
\usepackage{framed}
\usepackage{tabularx}

\newlength{\RoundedBoxWidth}
\newsavebox{\GrayRoundedBox}
\newenvironment{GrayBox}[1]%
   {\setlength{\RoundedBoxWidth}{.93\columnwidth}%
    \def\boxheading{#1}%
    \begin{lrbox}{\GrayRoundedBox}%
       \begin{minipage}{\RoundedBoxWidth}}%
   {   \end{minipage}
    \end{lrbox}
    \begin{center}
    \begin{tikzpicture}%
       \node(Text)[draw=black!20,fill=white,rounded corners,inner sep=2ex,text width=\RoundedBoxWidth]
             {\usebox{\GrayRoundedBox}};
        \coordinate(x) at (current bounding box.north west);
        \node [draw=white,rectangle,inner sep=3pt,anchor=north west,fill=white]
        at ($(x)+(6pt,.75em)$) {\boxheading};
    \end{tikzpicture}
    \end{center}}

\newenvironment{defproblemx}[1]{\noindent\ignorespaces%
                                \FrameSep=6pt%
                                \parindent=0pt%
                \begin{GrayBox}{#1}%
                \begin{tabular*}{\columnwidth}{!{\extracolsep{\fill}}@{\hspace{.1em}} >{\itshape} p{1.1cm} p{0.89\columnwidth} @{}}%
            }{
                 \vspace*{-1em}
                 \end{tabular*}%
                \end{GrayBox}%
                \ignorespacesafterend
            }

\newcommand{\defProblemQuestion}[3]{%
  \begin{defproblemx}{#1}
    Input: & #2 \\
    Question: & #3
  \end{defproblemx}
}

%% file: images/toolbox.tex
\def\isusedin(#1,#2,#3,#4,#5){
    \path [draw=#5, -latex',rounded corners=1pt] ($(#1.west)+(0,-0.04)$) -| ($(#1.west)+(-#3*0.05,-0.1)+(-#4*0.015,0)$) |- ($(#2.west)+(0,0.04)$);
}
\def\isusedfurther(#1,#2){
    \path [draw=#2, -latex',rounded corners=3pt] (#1.east) -- ($(#1.east)+(0.3,0)$);
}
\begin{tikzpicture}[
    node distance=1.0em,
    title/.style={font=\fontsize{6}{6}\color{black}\ttfamily, anchor=west}
    ]
    \def\shift{3pt}

    \def\copycolor{red!90!purple}
    \def\addcolor{ red!50!orange}
    \def\grpcolor{ red!10!orange}
    \def\propercolor{gray}
    \def\idcolor{  blue!90!purple}
    \def\lbcolor{  blue!60!purple}
    \def\starcolor{blue!40!purple}
    \def\pathcolor{blue!10!purple}

    \node (lb) at (0,0) [title] { \Cref{sec:lower_bounds} Lower bounds};
    \node (identification) [below=of lb.west, title, xshift=\shift] {\Cref{obs:identification} Vertex identification};
    \node (subst) [below=of identification.west, title] {\Cref{obs:subtitute_guards} Substitute guard};
    \node (cliquered)    [below=of subst.west, title] {\Cref{def:clique_reduction} Clique reduction};
    \node (lblem) [below=of cliquered.west, title, xshift=\shift] {\Cref{lem:lower_bound} Lower bound lemma};
    \node (ink)    [below=of lblem.west, title] {\Cref{lem:ink} Ink lemma};
    \node (lbleaf) [below=of ink.west, title, xshift=\shift] {\Cref{obs:lb_leaf} Leaf lower bound};
    \node (lbstar) [below=of lbleaf.west, title] {\Cref{obs:lb_star} Star lower bound};
    \node (lbpath) [below=of lbstar.west, title] {\Cref{obs:lb_path} Path lower bound};
    \node [draw=black, fit={(lb) (lbpath) (identification)}] {};
    \isusedin(identification,lblem,3,0,\idcolor)
    \isusedin(subst,lblem,2,0,\idcolor)
    \isusedin(ink,lbleaf,2,0,\lbcolor)
    \isusedin(lblem,lbleaf,4,0,\lbcolor)
    \isusedin(ink,lbstar,2,1,\starcolor)
    \isusedin(lblem,lbstar,4,1,\starcolor)
    \isusedin(ink,lbpath,2,2,\pathcolor)
    \isusedin(lblem,lbpath,4,2,\pathcolor)
    \isusedfurther(identification,\idcolor)
    \isusedfurther(lbleaf,\lbcolor)
    \isusedfurther(lbstar,\starcolor)
    \isusedfurther(lbpath,\pathcolor)

    \node (ub) at (6.2,0) [title] {\Cref{sec:upper_bounds} Upper bounds};
    \node (defs) [node distance=1.4em, below=of ub.west, title, xshift=\shift, align=left] {\Cref{def:states,def:interface,def:transition,def:partial,def:defended_graph,def:partial_subgraph}:\\ State, Movement, Interface, Transition,\\ (Partial) Defended graph and subgraph};
    \node (symmetry) [node distance=1.4em, below=of defs.west, title, xshift=1*\shift] {\Cref{prop:strategy_symmetry} Symmetry};
    \node (compatible) [below=of symmetry.west, title] {\Cref{def:compatible} Compatible};
    \node (cutting) [below=of compatible.west, title] {\Cref{def:cut} Cutting};
    \node (composing) [below=of cutting.west, title, xshift=0*\shift] {\Cref{def:composing} Composing};
    \node (compatiblecomp) [below=of composing.west, title, xshift=1*\shift] {\Cref{lem:compatible_composition} Composing compatible};
    \node (intequiv) [below=of compatiblecomp.west, title, xshift=0*\shift] {\Cref{def:interface_equivalent} Interface equivalent};
    \node (trancompatible) [below=of intequiv.west, title, xshift=1*\shift] {\Cref{lem:transfered_compatibility} Transferred compatibility};
    \node (expansion) [below=of trancompatible.west, title] {\Cref{def:expansion} Expansion};
    \node (eqconst) [below=of expansion.west, title] {\Cref{lem:equivalency_constant} Equivalency constant};
    \node (upperbound) [below=of eqconst.west, title, xshift=1*\shift] {\Cref{obs:upper_bound} Upper bound};
    \node [draw=black, fit={(ub) (upperbound) (defs)}] {};
    \isusedin(compatible,cutting,4,0,        green!50!black)
    \isusedin(compatible,compatiblecomp,4,1, green!55!black)
    \isusedin(composing,compatiblecomp,2,0,  green!55!black)
    \isusedin(compatible,trancompatible,4,2, green!60!black)
    \isusedin(intequiv,trancompatible,5,1,   green!60!black)
    \isusedin(compatiblecomp,expansion,4,0,  green!65!black)
    \isusedin(trancompatible,expansion,4,0,  green!65!black)
    \isusedin(intequiv,eqconst,5,2,          green!70!black)
    \isusedin(eqconst,upperbound,4,0,        green!75!black)
    \isusedin(expansion,upperbound,5,0,      green!75!black)

    \node (leafgroup) at (-0.7,-3.9) [title] {\Cref{sec:upper_bound_tools} Tools for altering strategies};
    \node (copy) [below=of leafgroup.west, title, xshift=1*\shift] {\Cref{def:copy} Cartesian product over subset};
    \node (leafaddition) [below=of copy.west, title, xshift=1*\shift] {\Cref{lem:extend_leaves} Leaves addition};
    \node (groupstate) [below=of leafaddition.west, title, xshift=1*\shift] {\Cref{def:group_state} Group state};
    \node [draw=black, fit={(leafgroup) (leafaddition) (groupstate) (copy)}] {};

    \isusedin(copy,leafaddition,2,0,\addcolor)
    \isusedin(copy,groupstate,2,1,\grpcolor)
    \isusedfurther(copy,\copycolor)
    \isusedfurther(leafaddition,\addcolor)
    \isusedfurther(groupstate,\grpcolor)
\end{tikzpicture}

%% file: images/overview.tex
\def\isusedin(#1,#2,#3,#4,#5){
    \path [draw=#5, -latex',rounded corners=1pt] ($(#1.east)+(0,-0.04)$) -| ($(#1.east)+(+#3*0.05,-0.2)+(+#4*0.015,0)$) |- ($(#2.east)+(0,0.04)$);
}
\def\isusedbetweenup(#1,#2,#3,#4,#5){
    \path [draw=#5, -latex',rounded corners=4pt] (#1.east) -| ($(#2.west)-(#3*0.1,0)+(-#4*0.015,-0.3)$) |- (#2.west);
}
\def\isusedbetweendown(#1,#2,#3,#4,#5){
    \path [draw=#5, -latex',rounded corners=4pt] (#1.east) -| ($(#2.west)-(#3*0.1,0)+(-#4*0.015,+0.3)$) |- (#2.west);
}
\def\isusedbetweendirect(#1,#2,#3,#4,#5){
    \path [draw=#5, -latex'] (#1.east) -- (#2.west);
}
\def\isusedbetweenstraight(#1,#2,#3,#4){
  \path [draw, -latex',rounded corners=4pt] (#1.east) -| ($(#2.west)-(#3*0.1,0)+(-#4*0.015,+0.0)$) |- (#2.west);
}
\def\isusedfurther(#1){
    \path [draw, -latex',rounded corners=4pt] (#1.east) -- ($(#1.east)+(0.3,0)$);
}
\begin{tikzpicture}[
    node distance=1.0em,
    title/.style={font=\fontsize{6}{6}\color{black}\ttfamily, anchor=west}
    ]
    \def\shift{5pt}

    \node (lb) at (0,0) [title] { \Cref{sec:lower_bounds} Lower bounds};
    \node (identification) [below=of lb.west, title, xshift=\shift] {\Cref{obs:identification} Vertex identification};
    \node (lbleaf) [below=of identification.west, title] {\Cref{obs:lb_leaf} Leaf lower bound};
    \node (lbstar) [below=of lbleaf.west, title] {\Cref{obs:lb_star} Star lower bound};
    \node (lbpath) [below=of lbstar.west, title] {\Cref{obs:lb_path} Path lower bound};
    \node [draw=black, fit={(lb) (lbpath) (identification)}] {};

    \node (properties) at (0,1.7) [title] {\Cref{sec:properties} Properties of cycle edges};
    \node (properstates) [below=of properties.west, title, xshift=1*\shift] {\Cref{prop:labelled_strategy} Proper labelled strategy};
    \node [draw=black, fit={(properties) (properstates)}] {};

    \node (leafgroup) at (0,4) [title] {\Cref{sec:upper_bound_tools} Tools for altering strategies};
    \node (copy) [below=of leafgroup.west, title, xshift=1*\shift] {\Cref{def:copy} Cartesian product over subset};
    \node (leafaddition) [below=of copy.west, title] {\Cref{lem:extend_leaves} Leaves addition};
    \node (groupstate) [below=of leafaddition.west, title] {\Cref{def:group_state} Group state};
    \node [draw=black, fit={(leafgroup) (leafaddition) (groupstate) (copy)}] {};

    \node (secleaves) at (8,3.1) [title] {\Cref{sec:reducing_trees} Leaf reductions};
    \node (redt1) [below=of secleaves.west, title, xshift=1*\shift] {\Cref{lem:reduction-t1} Reduction~\ref{reduction-t1}};
    \node (redt2) [below=of redt1.west, title] {\Cref{lem:reduction-t2} Reduction~\ref{reduction-t2}};
    \node (redt3) [below=of redt2.west, title] {\Cref{lem:reduction-t3} Reduction~\ref{reduction-t3}};
    \node (seccycle) [below=of redt3.west, title, xshift=-1*\shift] {\Cref{sec:cycle_reductions} Cycle reductions};
    \node (redc1) [below=of seccycle.west, title, xshift=1*\shift] {\Cref{lem:reduction-c1} Reduction~\ref{reduction-c1}};
    \node (redc23) [below=of redc1.west, title] {\Cref{lem:reduction-c23} Reductions~\ref{reduction-c2} and~\ref{reduction-c3}};
    \node (redc4) [below=of redc23.west, title] {\Cref{lem:reduction-c4} Reduction~\ref{reduction-c4}};
    \node (redc5) [below=of redc4.west, title] {\Cref{lem:reduction-c5} Reduction~\ref{reduction-c5}};
    \node (redc6) [below=of redc5.west, title] {\Cref{lem:reduction-c6} Reduction~\ref{reduction-c6}};
    \node (secconst) [below=of redc6.west, title, xshift=-1*\shift] {\Cref{sec:constant_reductions} Constant component reductions};
    \node (cactusmulti) [below=of secconst.west, title, xshift=1*\shift] {\Cref{def:cactus_multigraph} Cactus multigraph};
    \node (redm1) [below=of cactusmulti.west, title, xshift=1*\shift] {\Cref{lem:reduction-m1} Reduction~\ref{reduction-m1}};
    \node (redm2) [below=of redm1.west, title] {\Cref{lem:reduction-m2} Reduction~\ref{reduction-m2}};
    \node (redr1) [below=of redm2.west, title, xshift=-1*\shift] {\Cref{lem:reduction-r1} Reduction~\ref{reduction-r1}};
    \node (redr2) [below=of redr1.west, title] {\Cref{lem:reduction-r2} Reduction~\ref{reduction-r2}};
    \node (redr3) [below=of redr2.west, title] {\Cref{lem:reduction-r3} Reduction~\ref{reduction-r3}};
    \node (redr4) [below=of redr3.west, title] {\Cref{lem:reduction-r4} Reduction~\ref{reduction-r4}};
    \node (redr5) [below=of redr4.west, title] {\Cref{lem:reduction-r5} Reduction~\ref{reduction-r5}};
    \node [draw=black, fit={(secleaves) (cactusmulti) (redr5) (secconst)}] {};


    \def\copycolor{red!90!purple}
    \def\addcolor{ red!50!orange}
    \def\grpcolor{ red!10!orange}
    \def\propercolor{gray}
    \def\idcolor{  blue!90!purple}
    \def\lbcolor{  blue!60!purple}
    \def\starcolor{blue!40!purple}
    \def\pathcolor{blue!10!purple}

    \def\copydepth{15}
    \def\copymoredepth{16.72}
    \def\adddepth{17}
    \def\grpdepth{19}
    \def\properdepth{26}
    \def\iddepth{9}
    \def\lbdepth{7}
    \def\stardepth{5.5}
    \def\pathdepth{4.4}

    \isusedbetweendirect(properstates,redc1,\properdepth,0,\propercolor)
    \isusedbetweendown(  properstates,redc4,\properdepth,0,\propercolor)
    \isusedbetweendown(  properstates,redc5,\properdepth,1,\propercolor)

    \isusedbetweendown(groupstate,redt2,\grpdepth,0,\grpcolor)
    \isusedbetweendown(groupstate,redc23,\grpdepth,1,\grpcolor)

    \isusedbetweendown(leafaddition,redt3,\adddepth,0,\addcolor)
    \isusedbetweendown(leafaddition,redc5,\adddepth,1,\addcolor)
    \isusedbetweendown(leafaddition,redr2,\adddepth,2,\addcolor)
    \isusedbetweendown(leafaddition,redr3,\adddepth,3,\addcolor)
    \isusedbetweendown(leafaddition,redr4,\adddepth,4,\addcolor)

    \isusedbetweendown(copy,redt2,\copydepth,0,\copycolor)
    \isusedbetweendown(copy,redm1,\copymoredepth,1,\copycolor)
    \isusedbetweendown(copy,redm2,\copymoredepth,2,\copycolor)
    \isusedbetweendown(copy,redr3,\copydepth,3,\copycolor)
    \isusedbetweendown(copy,redr5,\copydepth,4,\copycolor)

    \isusedbetweenup(  identification,redt2,\iddepth,3,\idcolor)
    \isusedbetweendown(identification,redr1,\iddepth,0,\idcolor)
    \isusedbetweendown(identification,redr3,\iddepth,1,\idcolor)
    \isusedbetweendown(identification,redr4,\iddepth,2,\idcolor)
    \isusedbetweendown(identification,redr5,\iddepth,3,\idcolor)

    \isusedbetweenup(  lbleaf,redt1,\lbdepth,4,\lbcolor)
    \isusedbetweenup(  lbleaf,redt3,\lbdepth,3,\lbcolor)
    \isusedbetweenup(  lbleaf,redc1,\lbdepth,2,\lbcolor)
    \isusedbetweenup( lbleaf,redc23,\lbdepth,1,\lbcolor)
    \isusedbetweenup(  lbleaf,redc6,\lbdepth,0,\lbcolor)
    \isusedbetweendown(lbleaf,redr2,\lbdepth,0,\lbcolor)
    \isusedbetweendown(lbleaf,redr3,\lbdepth,1,\lbcolor)
    \isusedbetweendown(lbleaf,redr4,\lbdepth,2,\lbcolor)
    \isusedbetweendown(lbleaf,redr5,\lbdepth,3,\lbcolor)

    \isusedbetweenup(lbstar,redc5,\stardepth,1,\starcolor)
    \isusedbetweenup(lbstar,redc6,\stardepth,0,\starcolor)

    \isusedbetweenup(lbpath,redc4,\pathdepth,0,\pathcolor)

    \isusedin(redc4,redc5,5,0,green!60!black)

    \isusedin(redm2,redr1,2,0,green!50!black)
    \isusedin(redm2,redr2,2,1,green!55!black)
    \isusedin(redr1,redr2,6,0,green!55!black)
    \isusedin(redm1,redr3,5,0,green!60!black)
    \isusedin(redm2,redr3,2,2,green!60!black)
    \isusedin(redm1,redr4,5,1,green!65!black)
    \isusedin(redm2,redr4,2,3,green!65!black)
    \isusedin(redr3,redr4,6,0,green!65!black)
    \isusedin(redm1,redr5,5,2,green!70!black)
    \isusedin(redm2,redr5,2,4,green!70!black)
    \isusedin(redr4,redr5,6,0,green!70!black)
\end{tikzpicture}